\documentclass[twocolumn,pra,twocolumn,superscriptaddress]{revtex4}
\usepackage{color}
\usepackage{amsmath}
\usepackage{amssymb}
\usepackage{graphicx}
\usepackage{tikz}
\usepackage{xr}
\usepackage{xc}
\externalcitedocument[suppl-]{../supplap/suppl1127}

\usepackage{float}
\usepackage{bbm}
 \usepackage{mathrsfs}
\usepackage{epsfig} 
\usepackage{verbatim}
\usepackage{array}
\usepackage{setspace}

\usepackage[unicode=true,
 bookmarks=true,bookmarksnumbered=false,bookmarksopen=false,
 breaklinks=false,pdfborder={0 0 1},backref=false,colorlinks=true]
 {hyperref}
\hypersetup{
 linkcolor=magenta, urlcolor=blue, citecolor=blue, pdfstartview={FitH}, hyperfootnotes=false, unicode=true}
 
\usepackage{color}
\newenvironment{proof}{\noindent\textbf{Proof\ }}{\hfill\rule{3mm}{3mm}}
\newtheorem{lemma}{Lemma}
\newtheorem{theorem}{Theorem}
\newtheorem{corollary}{Corollary}

\makeatletter
\@ifundefined{textcolor}{}
{%
 \definecolor{BLACK}{gray}{0}
 \definecolor{WHITE}{gray}{1}
 \definecolor{RED}{rgb}{1,0,0}
 \definecolor{GREEN}{rgb}{0,1,0}
 \definecolor{BLUE}{rgb}{0,0,1}
 \definecolor{CYAN}{cmyk}{1,0,0,0}
 \definecolor{MAGENTA}{cmyk}{0,1,0,0}
 \definecolor{YELLOW}{cmyk}{0,0,1,0}
}

\usepackage{amsfonts}\usepackage{tabularx}\usepackage{dcolumn}\usepackage{bm}\usepackage{graphicx}\usepackage{epstopdf}

\hypersetup{urlcolor=blue}

\makeatother

\usepackage{algorithm}
\usepackage{algorithmic}

\begin{document}

\widetext

\title{Minimal nonorthogonal gate decomposition for qubits with limited control}

\author{Xiao-Ming Zhang}
\affiliation{Shenzhen Institute for Quantum Science and Engineering and Department of Physics,
Southern University of Science and Technology, Shenzhen 518055, China}
\affiliation{Department of Physics, City University of Hong Kong, Tat Chee Avenue, Kowloon, Hong Kong SAR, China, and City University of Hong Kong Shenzhen Research Institute, Shenzhen, Guangdong 518057, China}

\author{Jianan Li}
\affiliation{Shenzhen Institute for Quantum Science and Engineering and Department of Physics,
Southern University of Science and Technology, Shenzhen 518055, China}

\author{Xin Wang}
\email{x.wang@cityu.edu.hk}
\affiliation{Department of Physics, City University of Hong Kong, Tat Chee Avenue, Kowloon, Hong Kong SAR, China, and City University of Hong Kong Shenzhen Research Institute, Shenzhen, Guangdong 518057, China}
\author{Man-Hong Yung}
\email{yung@sustc.edu.cn}
\affiliation{Shenzhen Institute for Quantum Science and Engineering and Department of Physics,
Southern University of Science and Technology, Shenzhen 518055, China}
\affiliation{Shenzhen Key Laboratory of Quantum Science and Engineering, Southern University of Science and Technology, Shenzhen, 518055, China}
\affiliation{Central Research Institute, Huawei Technologies, Shenzhen, 518129, China}

\date{\today}
\begin{abstract}
In quantum control theory, a question of fundamental and practical interest is how an arbitrary unitary transformation can be decomposed into minimum number of elementary rotations for implementation, subject to various physical constraints. Examples include the singlet-triplet (ST) and exchange-only (EO) qubits in quantum-dot systems, and gate construction in the Solovay-Kitaev algorithm. For two important scenarios, we present complete solutions to the problems of optimal decomposition of single-qubit unitary gates with non-orthogonal rotations. For each unitary gate, the criteria for determining the minimal number of pieces is given,  the explicit gate construction procedure, as well as a computer code for practical uses. Our results include an analytic explanation to the four-gate decomposition of EO qubits, previously determined numerically by Divincenzo et al [Nature, 408, 339 (2000)]. Furthermore, compared with the approaches of Ramon sequence and its variant [Phys. Rev. Lett., 118, 216802 (2017)], our method can reduce about 50\% of gate time for ST qubits. Finally, our approach can be extended to solve the problem of optimal control of topological qubits, where gate construction is achieved through
the braiding operations. 
\end{abstract}
\maketitle

A universal gate set for quantum computation can be constructed by any two-qubit entangling gate, together with arbitrary single-qubit gates~\cite{Nielsen.00}. In the laboratory, elementary single-qubit gates are normally constructed by switching on and off an external field at certain times (i.e., a square pulse), resulting in a rotation of a Bloch vector along certain axis of the Bloch sphere. The question is, {\it how to optimize the use of these elementary rotations to form  arbitrary single-qubit gates?} This question becomes crucial for quantum platforms where controls are {\it limited}~\cite{Petta.05,Maune.12}. Consequently, a general rotation needs to be decomposed into a sequence of elementary rotations around non-parallel axes. In fact, this ``piecewise'' decomposition of general operations has inspired the development of composite pulses, which play an important role in quantum control on various types of qubits~\cite{Wimperis.94,Cummins.03,Wang.12,Bando.13,Kestner.13,Kosut.13,Wang.14}. 

Typically, one would like to reduce the complexity of gates: a long sequence of elementary gates implies the need of frequent switching of the applied field. Therefore, a minimal decomposition of single-qubit gates is of practical and fundamental interest in quantum computing. For the cases where the available elementary controls are rotations around two {\it orthogonal} axes, it is well known that arbitrary rotations can be constructed with three pieces~\cite{Nielsen.00}, for example the $x$-$z$-$x$ sequence~\cite{gatenote}. 

However, in many systems, the available elementary rotations are non-orthogonal. Take singlet-triplet (ST) qubit as an example~\cite{Petta.05}, the $x$-rotation can be achieved via a magnetic field gradient \cite{Foletti.09,Maune.12,Petersen.13, Wu.14}, but a pure $z$-rotation is hardly achievable as the magnetic field gradient has to be completely turned off during execution of a gate, which is impractical unless the micromagnet is applied~\cite{Brunner.11}. As another example, control of an exchange-only (EO) qubit \cite{Divincenzo.00,Laird.10} is only available via two rotation axes 120$^\circ$ apart from each other. 

In the literature of quantum dots, much effort has been made to optimize gate sequences involving non-orthogonal axes \cite{Divincenzo.00,Hanson.07,Ramon.11,Zhang.16,Zhang.17,Shim.13,Throckmorton.17}. If the rotation axes along $\hat{x}$ and $\hat{x}+\hat{z}$ are both available, a Hadamard gate can convert an $x$-rotation to a $z$-rotation, providing an $x$-Hadamard-$x$-Hadamard-$x$ sequence~\cite{Wang.12,Wang.14}. Moreover, if the angle between the two available axes (denoted as $\hat{x}$ and $\hat{m}$) is greater than $45^\circ$, Hadamard gate can be replaced by the rotation around $\hat{m}$ to reduce the gate time~\cite{Zhang.17,Throckmorton.17}. But the number of pieces is still unchanged.
It remains an outstanding problem whether more efficient decompositions with non-orthogonal axes is possible. 

In an early study of the EO qubit, Divincenzo et al. {\it numerically} found that four-piece sequences can be constructed for almost all quantum gates~\cite{Divincenzo.00}, but no analytical explanation was given. Furthermore, in applying the Solovay-Kitaev theorem \cite{Nielsen.00}, it was believed that an arbitrary gate can be decomposed into three pieces \cite{Nielsen.00,Kaye.07}, but the problem turns out to be far more complicated.

Here, for two typical scenarios: the elementary rotation axes are fixed along two directions, or can vary in a range of a plane, we present complete solutions to the problem of  minimal decomposition of single-qubit transformation. We determine the minimum number of pieces for any given unitary transformation, and the explicit procedures in constructing the minimum decomposition are also provided as computer code~\cite{sm} (Appendix~\ref{sec:code}). Furthermore, we obtain the minimum number of pieces for all possible unitary, which turn out to have the same expression for both scenarios. For applications, we demonstrate how minimal decomposition can be implemented in ST qubit systems, which can improve the gate time and robustness under real experimental circumstances. 

\textit{\it Zeeman type qubits}---We consider a Zeeman type Hamiltonian
\begin{equation}
H_{\rm{ZM}}=h \, \sigma_x+J(t) \, \sigma_z,\label{eq:st}
\end{equation}
where $h$ is a constant, $J(t)$ can vary with time, and $\sigma_x$, $\sigma_z$ are Pauli matrices. Eq.~\eqref{eq:st} can represent a two-level system with fixed energy gap $h$ under an external control field $J(t)~$\cite{Nartinis.05,Greilich.09,Poem.11}. Besides, it can also represents the ST qubits of quantum-dot systems~\cite{Petta.05,Hanson.07,Foletti.09,Ramon.11,Maune.12,Barnes.12,Wang.12,Petersen.13,Wu.14,Zhang.17}, where $h$ and $J(t)$ represents the magnetic field gradient and exchange interaction respectively. The value of $J(t)$ should be bounded within a certain range, $0\leqslant J(t)\leqslant J_{\max}$, in order to satisfy the requirement $B\gg J_{\max}$ ($B$ the average magnetic field strength) that ensure other energy levels far away from the two we concerned~\cite{Petta.05, Hanson.07,Wu.14}. If we let $\Theta=\arctan(J_{\max}/h)$, the available rotations is given by $\mathcal{G}_{ZM}=\{R(\hat{n},\phi)|\hat{n}=(\sin\theta,0,\cos\theta),\theta\in[\frac{\pi}{2}-\Theta,\frac{\pi}{2}]\}$.

\textit{\it Exchange-only (EO) qubits.}---The EO qubit is constructed by a coupled triple-quantum-dot system. Assuming a homogeneous magnetic field, the Hamiltonian in this subspace can be written as~\cite{Divincenzo.00,Laird.10,Gaudreau.12,Medford.13,Zhang.16}
\begin{equation}
H_{\rm{EO}}=J_{23} \ \sigma_z-J_{12} \ (\frac{1}{2}\sigma_z-\frac{\sqrt{3}}{2}\sigma_x), \label{eq:eo}
\end{equation}
 where $J_{12}\geqslant 0$ and $J_{23}\geqslant 0$ are coupling constants between the neighboring  dots. However, it remains an experimental challenge to simultaneous apply both coupling, which means that either $J_{23}$ or $J_{12}$ should be non-zero at each moment of time. In other words, only elementary rotations around $\hat{z}$ or another axis $\sqrt{3}\hat{x}/2-\hat{z}/2$ can be applied, i.e., $\mathcal{G}_{EO}=\left\{R(\hat{n},\phi)|\hat{n}=\hat{z} \; \text{or}\; \hat{n}=\sqrt{3}\hat{x}/2-\hat{z}/2 \right\}$. 
 
 Motivated by the systems described above, we will consider two types of models. They can include but are not limited to the systems represented by Eq.~\eqref{eq:st} and Eq.~\eqref{eq:eo}.

{\it Definitions}---A single-qubit rotation, $R(\hat{n},\phi)$, around the axis $\hat{n}=(\sin\theta\cos\psi,\sin\theta\sin\psi,\cos\theta)$ for an angle $\phi\in[0,4\pi)$, can be generically described by,
\begin{equation}
R(\hat{n},\phi) \equiv \exp [-i (\bm{\sigma}\cdot\hat{n}) \phi/2] \ ,
\end{equation}
where $\bm{\sigma}\equiv(\sigma_x,\sigma_y,\sigma_z)$ contains the Pauli matrices. We are interested in how a unitary gate $U(\theta,\psi,\phi)$ (up to an overall phase factor) can be {\it minimally} decomposed into a sequence of elementary rotations , $R_i=R(\hat{n}_i,\phi_i)$ in a given set $\mathcal{G} = \{ R(\hat{n}_i,\phi_i) \}$ limited by physical constraints. 

For convenience, we define the $p$-power of a set $\mathcal{G}$ to contain all combinations of products of $p$ elementary rotations, i.e., ${\mathcal{G}^p} \equiv \{ \prod\nolimits_i^p {{R_i}} |{R_i} \in \mathcal{G}\} $. Our task is to solve the following decomposition:
\begin{equation}
U(\theta,\psi,\phi)=\prod^{p}_{i=1}R_i \in \mathcal{G}^p, \label{general}
\end{equation}
subject to the condition, $\hat{n}_i\neq\hat{n}_{i+1}$. Here $p$ is referred to as ``number of pieces''. Of course, for each $U$ the solution of $p(U)$ satisfying the decomposition is not unique; in fact, there are infinitely many possible solutions.

The goal of this work is to determine the minimum value $p_{\min} (U)$ for any given unitary transformation $U(\theta,\psi,\phi)$.

\begin{figure}[t]
\includegraphics[width=0.9 \columnwidth]{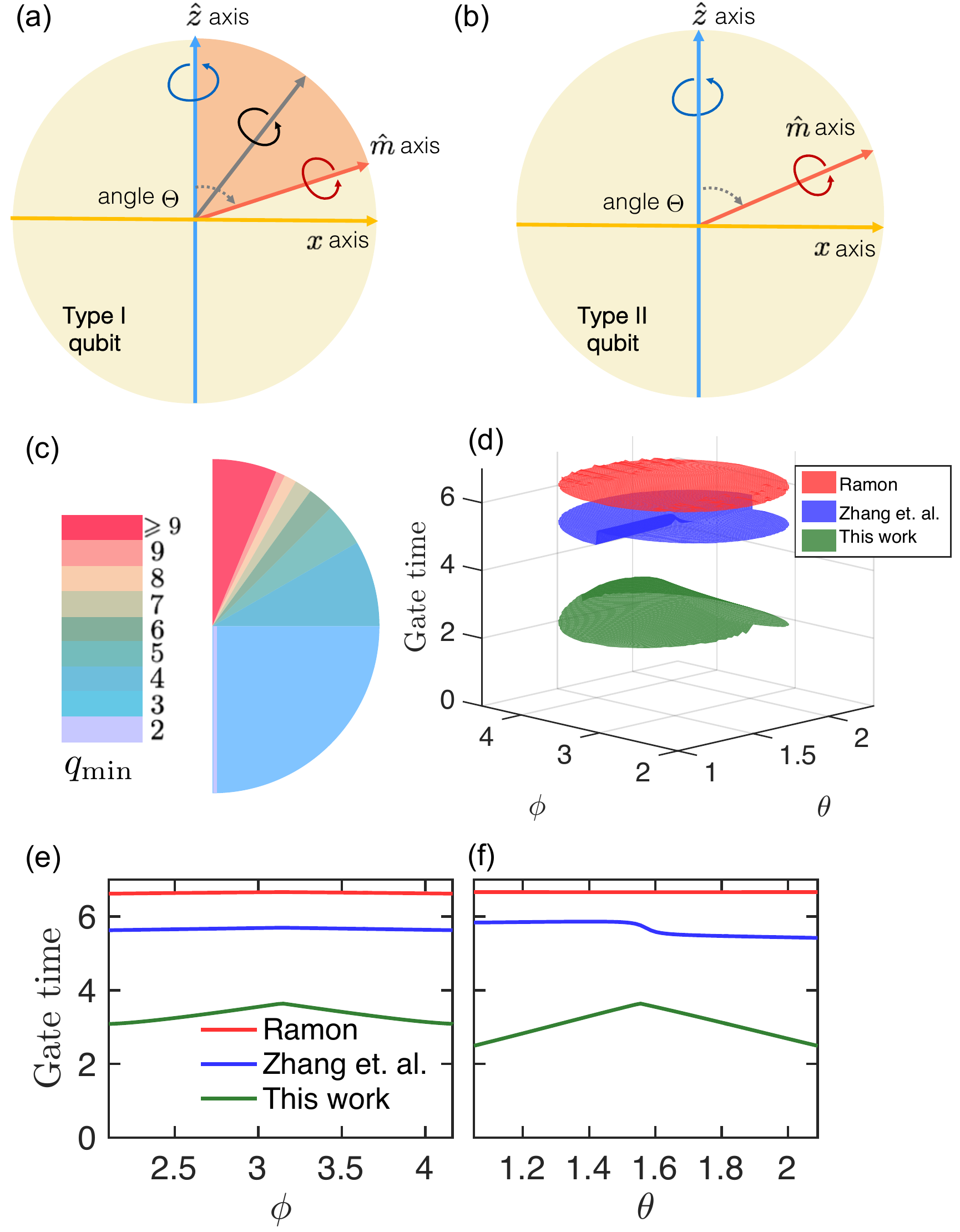}
\caption{ (a) Schematic description of Type-I qubits. Rotation axes are allowed to be chosen freely in a range (orange) bounded by $\hat{z}$ and $\hat{m}$ axes. (b) Schematic description of Type-II qubits. Rotation axes are fixed to be either $\hat{z}$ or $\hat{m}$. (c) The minimum number of pieces for {\it all possible} rotations ($q_{\min}$) for Type-I qubits. The color scale represents different $q_{\min}$ when $\hat{m}$ lies in the corresponding area. (d)-(f) are gate time comparison for ST qubits. Red (upper surface or lines): five-piece ''Ramon" sequence~\cite{Ramon.11}; blue (middle surface or lines): revised schemes proposed by Zhang et al.~\cite{Zhang.17}; green (lower surface or lines): minimal decomposition proposed in this work. Target gates are $U(\theta,\pi/2,\phi)$, and we set $J_{\max} = 30 h$. (e) and (f) are 2D cut for (d) with $\theta=\pi/2$ and $\phi=\pi$ respectively.
\label{fig:rot}}
\end{figure}

\begin{figure}[t]
\includegraphics[width=0.8 \columnwidth]{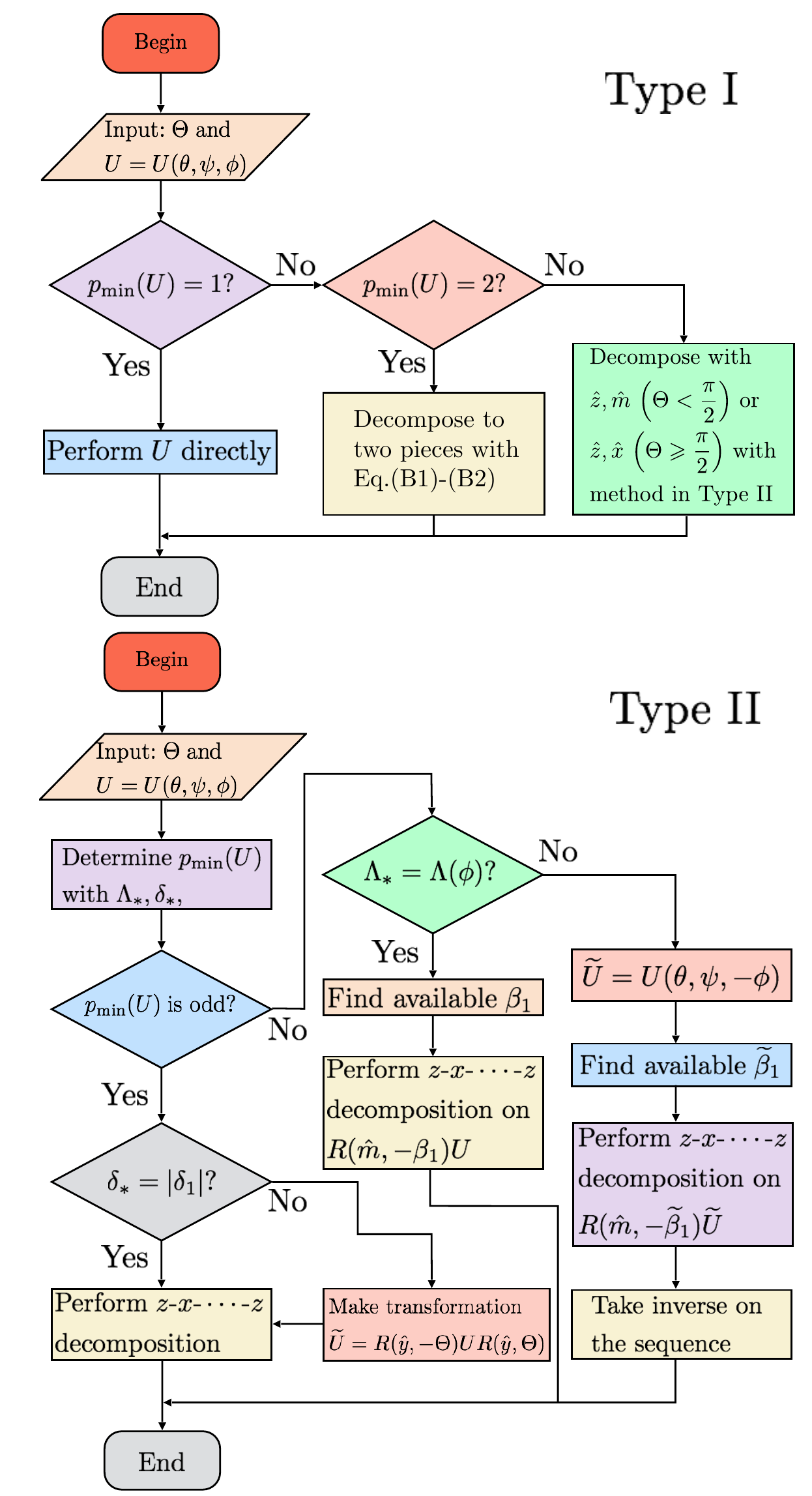}
\caption{ Flow chart for constructing the minimum decomposition sequences for Type-I qubits and Type-II qubits.
\label{fig:flow}}
\end{figure}

{\it Type-I qubits}---Here, we consider a scenario where the rotation axes are allowed to vary in a limited range of a plane. We suppose the range is enclosed by the boundary rotation axes denoted by $\hat{z} $ and $\hat{m}=(\sin\Theta,0,\cos\Theta) $; the angle between the boundary axes are given by $\Theta=\arccos \hat{z}^T\hat{m}\in(0,\pi]$ [see Fig.~\ref{fig:rot} (a)]. We define the set containing all possible elementary rotations by $\mathcal{G}_\xi\equiv\{R(\hat{n},\phi)|\hat{n}=(\sin\theta,0,\cos\theta)\}$. Furthermore, the boundary of $\mathcal{G}_\xi$ is given by the joint set of the rotations: 
\begin{equation}\label{gbugzgm}
\mathcal{G}_b\equiv\mathcal{G}_z\cup\mathcal{G}_m \ ,
\end{equation}
where $\mathcal{G}_z\equiv\{R(\hat{z}, \phi)\}$ and $\mathcal{G}_m\equiv\{R(\hat{m}, \phi)\}$.

 The reason why we use $\mathcal{G}_\xi$ to describe Type-I qubits is that the derivation and the expression of the results can be simpler and more elegant. Elementary rotation set $\mathcal{G}_{ZM}$ for Eq.~\eqref{eq:st} can be mapped to $\mathcal{G}_\xi$ with appropriate coordinate transformation~\cite{rotnote}. 

Below, we will present all the cases where Eq.\eqref{general} can be satisfied for $\mathcal{G}= \mathcal{G}_\xi$ and $\Theta\in(0,\pi]$ with a certain value of $p$ (see proofs in Appendix~\ref{app:B}). For $p=1$, Eq.~\eqref{general} can be satisfied, \textit{if and only if} one of the following conditions are satisfied: (i) $\theta\in [0,\Theta] $ and $\psi=0$, (ii) $\phi\in\{0, 2\pi\}$ (No rotation is applied, or adding a trivial global phase) or (iii) $\theta=0$ (around $\hat{z}$). For $p=2$, Eq.~\eqref{general} can be satisfied, \textit{if and only if} one of the following conditions are satisfied (i) $\phi\in\{0, 2\pi\}$, (ii) $\theta=0$, (iii) $\max\left\{\cot\xi_{\pm}\right\} \geqslant \cot\Theta$, with $\cot\xi_{\pm}= (\pm s_\psi c_{\phi/2} + c_\psi s_{\phi/2} c_\theta)/(s_{\phi/2}s_\theta)$, where we defined ${c_x} \equiv \cos x$ and ${s_x} \equiv \sin x$. (iv) $\Theta=\pi$. In case (iv), the rotation axes can be chosen freely in the entire $x$-$z$ plane; two pieces are sufficient, which is consistent with the result in Ref.~\cite{Shim.13}.

For $p\geqslant 3$, the results are summarized in Theorem.~\ref{th:1}. 
\begin{theorem}[Bulk-to-boundary mapping]\label{th:1}
(i) For $ 0 < \Theta < \pi/2$, if a unitary gate $U$ can be decomposed to $p \ge 3$ pieces, $U\in\mathcal{G}_{\xi}^p$, it can always be decomposed into $p$ pieces with rotation axes at the boundary, i.e.,
 \begin{equation}
\mathcal{G}_{\xi}^p=\mathcal{G}_{b}^p.\label{equ}
\end{equation}
(ii) for $\Theta \ge \pi/2$, one can always apply the orthogonal $z$-$x$-$z$ decomposition for any single-qubit unitary gate with $p=3$ pieces. 
\end{theorem}
When (i) $\Theta\in(0,\pi/2)$, it can be reduced to the Type-II with same $\Theta$ apart, so the existence of $p$-piece decomposition is determined by Eq.~\eqref{eq:odd} and \eqref{eq:even}, which will be described below; when (ii) $\Theta\in[\pi/2,\pi)$, decomposition with $p\geqslant 3$ pieces always exist, which  is obvious.

 \begin{figure}[t]
\includegraphics[width=1\columnwidth]{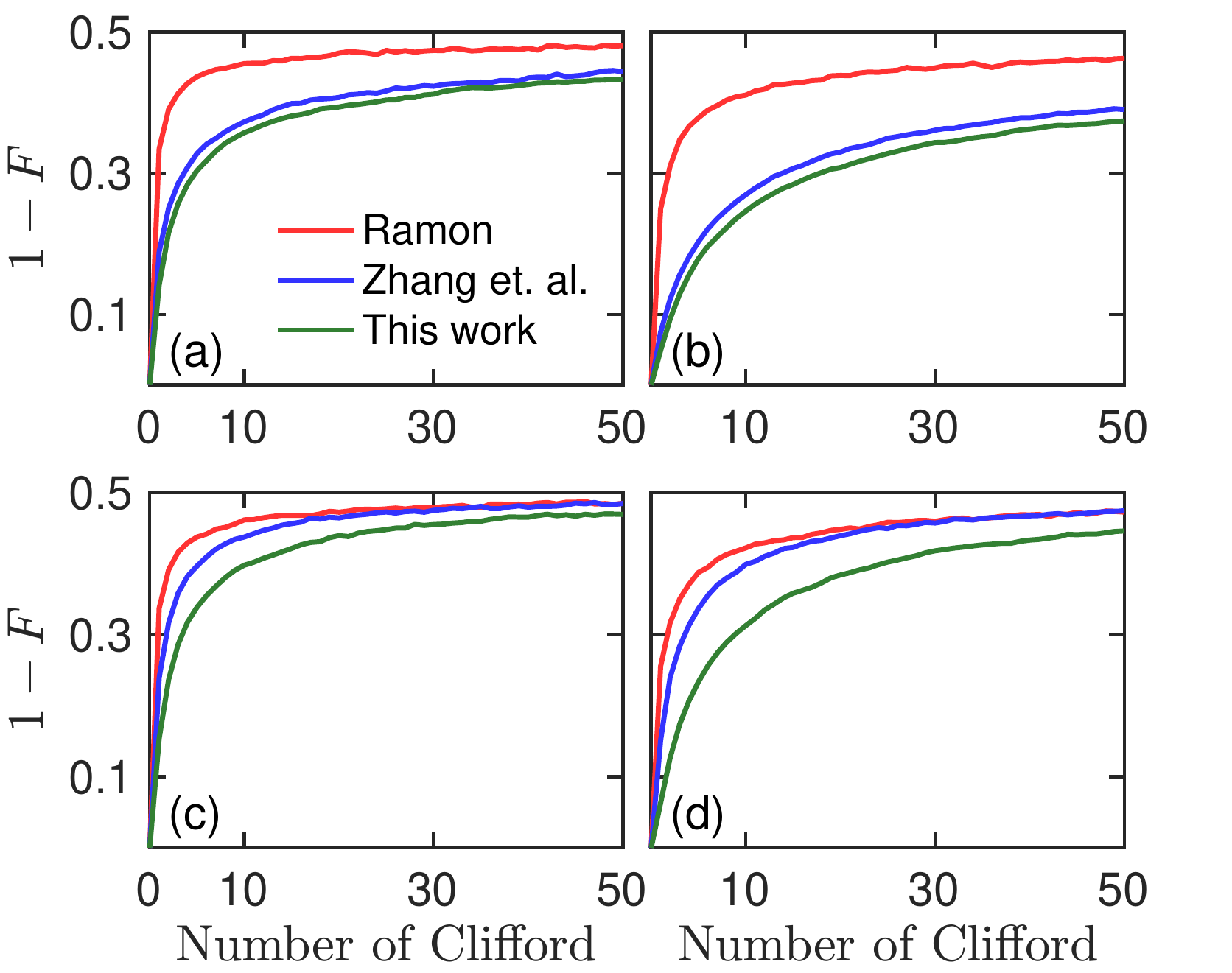}
\caption{ Randomized Benchmarking for different decomposition schemes for ST qubits. Red (upper) lines: ``Ramon" sequence \cite{Ramon.11}; blue (middle) lines: revised scheme in \cite{Zhang.17}; green (lower) lines: minimal decomposition proposed in this work. We set $J_{\max}=30$ and $h=1$; $\sigma_J/J=0.00426$ (barrier control) for (a), (b) and $\sigma_J/J=0.0563$ (tilt control) for (c), (d); $\sigma_h/h=0.575$ for (a), (c) and  $\sigma_h/h=0.288$ for (b), (d).}
\label{fig:rb}
\end{figure}

{\it Type-II qubits}---Then, we consider when only elementary rotations with two fixed axes are allowed, for example, $\hat{z}$ and $\hat{m}$, where the angle between them is given by $\Theta=\arccos \hat{z}^T\hat{m}\in(0,\pi/2]$\cite{rotnote}. The set containing all elementary rotations is given by $\mathcal{G}_b$ [see Eq.~(\ref{gbugzgm}) and Fig.~\ref{fig:rot} (b)].  For  any target $U$ and angle $\Theta$, we have solved the problem of minimal gate decomposition, in terms of a pair of inequalities:

(i) For the odd-piece decomposition, i.e.,~$p= 2l-1$, for some $l\in \mathbb{Z}^{+}$, the decomposition in Eq.~\eqref{general} can be satisfied for a given rotation {\it if and only if}
\begin{equation}
\delta _* \leqslant \Theta(l - 1) \label{eq:odd} ,
\end{equation}
where the value of ${\delta _{*}} \equiv \min \left\{ {|\delta_1 (\theta ,\phi )|,|\delta_2(\theta ,\psi ,\phi ,\Theta )|} \right\}$ is taken to be the minimum value between $\delta_1 (\theta ,\phi ) \equiv \sin^{ - 1}({s_\theta }{s_{\phi /2}})$ and $\delta_2(\theta ,\psi ,\phi ,\Theta ) \equiv {\sin ^{ - 1}}[{s_{\phi /2}}\sqrt {{{({c_\Theta }{c_\psi }{s_\theta } - {c_\theta }{s_\Theta })}^2} + {{({s_\theta }{s_\psi })}^2}} ]$. 

Furthermore, the form of $\delta_*$ determines the resulting sequence. When $\delta_* = \delta_1(\theta,\phi)$, Eq.~\eqref{general} can be constructed by the sequence: $U(\theta,\psi,\phi)=R(\hat{z},*)R(\hat{m},*)R(\hat{z},*)\cdots$; when $\delta_*=\delta_2(\theta,\psi,\phi,\Theta)$, Eq.~\eqref{general} can be constructed in the form of $U(\theta,\psi,\phi)=R(\hat{m},*)R(\hat{z},*)R(\hat{m},*)\cdots$.

(ii) For the even-piece decomposition where $p= 2l$, the decomposition in Eq.~\eqref{general} can be satisfied for a given rotation, {\it if and only if}
\begin{equation}
 \Lambda_* \leqslant \Theta(l - 1) \label{eq:even} ,
\end{equation}
where $\Lambda_* \equiv \min\{ \Lambda\left( \theta, \psi, \phi, \Theta \right), \Lambda\left( \theta, \psi, -\phi, \Theta \right) \}$ and $\Lambda(\theta,\psi,\phi,\Theta) \equiv \sin^{-1} \sqrt{ \frac{A + B}{2} - \sqrt{C^2 + \frac{(B - A)^2}{4}} }$. The other variables are defined as follows: $A \equiv {\left( {{c_\psi }{c_\Theta }{s_\theta }{s_{\phi /2}} - {s_\Theta }{c_\theta }{s_{\phi /2}}} \right)^2} + {\left( {{s_\psi }{c_\Theta }{s_\theta }{s_{\phi /2}} - {s_\Theta }{c_{\phi /2}}} \right)^2}$, $B \equiv \left( s_\theta s_{\phi /2} \right)^2$, and $C \equiv
{s_\Theta }{s_\theta }{s_{\phi /2}}({s_\psi }{s_{\phi /2}}{c_\theta } - {c_\psi }{c_{\phi /2}})$. 
The constructions of the decompositions are discussed in Appendix~\ref{sec:fix}.

\textit{Minimum number of pieces for all possible $U$}---From the experimental point of view, it is of interest to determine the optimal number of pieces applicable for {\it all possible} unitary transformations, i.e., 
\begin{equation}
{q_{\min }} \equiv \mathop {\max }\limits_U  \ {p_{\min }}\left( U \right)  \ .
\end{equation}
In principle, the values of ${q_{\min }}$ for Type-I and Type-II qubits can be different, as they are subject to different physical constraints. However, as shown below, they are actually identical. 

 It is known that~\cite{Lowenthal.71,Hamada.14}, for Type-II qubits, all rotations can be decomposed to $p\geqslant3$ pieces, \textit{if and only if} $\Theta\geqslant\pi/(p-1)$ (see Appendix~\ref{sec:min} for alternative proof), which implies
\begin{equation}
{q_{\min }} = \left\lceil {\frac{\pi }{\Theta }} \right\rceil  + 1. \label{all}
\end{equation}

In particular, for EO qubits, where two available rotation axes are fixed with relative angle $\Theta=\pi/3$, our results imply that the minimum number of pieces is given by $q_{\min}=4$, which represents an analytic explanation to the numerical results obtained by Divincenzo et al in 2000~\cite{Divincenzo.00}.  

For Type-II qubits, we know from Theorem.~\ref{th:1} that when $\Theta\in(0,\pi)$, $q_{\min}$ is the same as Type-I. When $\Theta=\pi$, criteria (iv) for $p=2$ indicates that $q_{\min}=2$, so Eq.~\eqref{all} also holds for Type-I qubits. An illustration of $q_{\min}$ is given in Fig.~\ref{fig:rot}~(c).

\textit{Improving  ST qubits control}---  By minimizing the number of pieces, one can reduce the error introduced by imperfect control field switching. In the existing ST qubits literatures~\cite{Ramon.11,Zhang.17,Throckmorton.17}, the single-qubit gates are typically decomposed into five or more pieces; our results show that as long as $\Theta\geqslant \pi/3$, all target rotation can be decomposed to four or even less number of pieces [see Eq.~\eqref{all} below]. Specifically, when $J_{\max}=30h$ which is a typical experiment value \cite{Martins.16,Reed.16}, we have found that for the set of 24 Clifford gates, $10$ gates can be realized with $p_{\min} (U)\leqslant2$, and $13$ gates with $p_{\min} (U)=3$. 

However, operations with $J\simeq0$ is slow and may suffer from severe nuclear noise \cite{Zhang.17}, solely minimizing the number of pieces is not optimum. To avoid operations with $J\simeq0$ while using as small number of pieces as possible, we propose the following decomposition strategy. Given maximum coupling strength $J_{\max}$,  we restrict $J\in[J_{\max},J_{\min}]$, where $ J_{\min}/h=\tan(\arctan J_{\max}/h-\pi/3)$. 
This ensures the axes can vary in a range with $\Theta=\pi/3$, and $p_{\min}\leqslant 4$. For a given target rotation, we decompose it with $p=1$ or $p=2$ if such solutions exist. Otherwise, the decomposition with $p=3,4$ is realized with fixed axes at the boundary corresponding to $J=J_{\max}$ and $J=J_{\min}$ (and gate time are optimized).

We compare our minimal decomposition scheme to (i) five-piece Ramon sequence~\cite{Ramon.11}, realized by alternating couplings between $J_{\max}$ and $0$, and (ii) an alternative scheme~\cite{Zhang.17} designed for avoiding operations for the $J=0$ case.  Remarkably, the average gate time for Clifford gate is $46\%$ and $71\%$ shorter compared to~\cite{Zhang.17} and~\cite{Ramon.11} respectively. Moreover, Fig.~\ref{fig:rot} (d)-(f) shows the comparison of gate time for several target unitary gates with $p_{\min}(U)=4$ [$U(\theta,\pi/2,\phi)$ of different values of $\theta$ and $\phi$]. For these family of gates, our minimal decomposition scheme has on average $48\%$ and $56\%$ shorter gate time relative to~\cite{Zhang.17} and~\cite{Ramon.11} respectively.

To further study the robustness, we perform randomized benchmarking with Gaussian static noise. The nuclear spin noise $\delta h$ are drawn from $\mathcal{N}(0,\sigma_{h}^2)$, and the charge noise are drawn from $\mathcal{N}(0,\sigma_{J}^2/J^2)$.  Fig.~\ref{fig:rb} shows the average gate fidelity for different values of $\sigma_J$ and $\sigma_h$  corresponding to barrier control [Fig.~\ref{fig:rb} (a), (b)] or tilt control [Fig.~\ref{fig:rb} (c), (d)] of GaAs quantum dots~\cite{Martins.16}. The results show that our minimal decomposition scheme can provide improvement in the robustness for real experimental circumstance. We note that it is also possible to allow larger $p_{\min}$ and let the axes farther away from $x$-axis, or even allow continuous tuning of the control fields~\cite{Castelano.18}. Finding optimal control scheme remains an open question, which beyond the scope of this work.

Furthermore, the control scheme should be designed on a case-by-case basis. For example, qubits hosted in isotropic purified material~\cite{Yoneda.18,Muhonen.15,Chan.18} has negligible nuclear noise, so the rotations along $x$-axis are no longer unfavored. In this case, the error introduced by imperfect control field switching becomes important.

To conclude, we have studied the minimal decomposition for two types of qubits: rotation axes are restricted in a range of a plane (Type-I), and rotation axes are fixed at two directions (Type-II). We also present an explicit procedure for minimally applying the elementary gates for an arbitrary single-qubit transformation. Furthermore, we discuss the implications of minimal decomposition for ST qubit, providing numerical evidences showing the effectiveness and robustness of our decomposition. Finally, we provide a code online~\cite{sm} for experimentalists, who just need to input a target rotation; the code will generate the explicit minimal decomposition.
The combination of our work with dynamical decoupling~\cite{Wimperis.94,Wang.12,Wang.14} or geometric control~\cite{Duan.01,Liu.18,Yan.18} can be interesting in the future.

\section*{Acknowledgements}
We thank Chengxian Zhang for helpful discussion. This work is supported by the National Natural Science Foundation of China (No. 11875160, No. 11604277), the NSFC
Guangdong Joint Fund (U1801661), the Guangdong Innovative and Entrepreneurial Research Team Program (No.~2016ZT06D348), the Research Grants Council of the Hong Kong Special Administrative Region, China (No.~CityU 21300116, CityU 11303617, CityU 11304018),  Natural Science Foundation of Guangdong Province (2017B030308003), and the Science, Technology and Innovation Commission of Shenzhen Municipality (JCYJ20170412152620376, JCYJ20170817105046702, ZDSYS201703031659262).
\begin{appendix}

\section{Definition}
To facilitate the discussions, for $R(\hat{n},\phi)$ with $\hat{n}=(\sin\theta\cos\psi,\sin\theta\sin\psi,\cos\theta)$, we parametrize it as:
\begin{align}
 R(\hat{n},\phi) &\equiv R(\theta,\psi,\phi) \notag\\
& \equiv \left[\begin{array}{cc}
	\cos \frac{\phi}{2} - i \sin \frac{\phi}{2} \cos \theta & - i \sin
	\frac{\phi}{2} \sin \theta e^{- i \psi}\\
	- i \sin \frac{\phi}{2} \sin \theta e^{i \psi} & \cos \frac{\phi}{2} + i
	\sin \frac{\phi}{2} \cos \theta
	\end{array}\right],  \label{SU2}
\end{align}
where $\theta\in[0,\pi)$, $\psi\in[0,\pi)$, $\phi\in[0,4\pi)$ unless otherwise specified. For clarity, we represent all \textit{target} unitary transformation as $U(\theta,\psi,\phi)\equiv R(\theta,\psi,\phi)$. Inversely, given
$R \left( \theta,\psi,\phi \right)= \left[\begin{array}{cc}
	e_{11} &e_{12}\\
	e_{21} &e_{22}
	\end{array}\right]$,
one can calculate angles as follows, which are important for the actual construction of the decomposition:

If $\text{Re} (e_{21})>0$,
\begin{subequations}
\begin{align}
\psi&=\text{Arg}(ie_{21}),\\
\phi&=2\arccos[\text{Re}(e_{11})],\\
\theta&=\left\{ 
\begin{array}{ll}
\arccos\frac{-\text{Im}(e_{11})}{\sin(\phi/2)},&\phi\neq0, 2\pi\\
0,&\phi=0\ \mathrm{ or }\ 2\pi
\end{array}
\right.
\end{align}\label{eq2}
\end{subequations}
If $\text{Re} (e_{21})<0$,
\begin{subequations}
\begin{align}
\psi&=\text{Arg}(-ie_{21}),\\
\phi&=4\pi-2\arccos[\text{Re}(e_{11})],\\
\theta&=\left\{ 
\begin{array}{ll}
\arccos\frac{-\text{Im}(e_{11})}{\sin(\phi/2)},&\phi\neq0,2\pi\\
0,&\phi=0\ \mathrm{ or }\ 2\pi
\end{array}
\right.
\end{align}\label{eq3}
\end{subequations}
If $\text{Re} (e_{21})=0$, 
\begin{subequations}
\begin{align}
\psi&=0,\\
\theta&=\text{Arg}[-\text{Im}(e_{11})-e_{21}  ],\\
\phi&=\left\{   
\begin{array}{ll}
 2\text{Arg}\left[\text{Re}(e_{11})-i\frac{\text{Im}(e_{11})}{\cos\theta} \right],&\theta\neq\pi/2  \\
         2\text{Arg}(e_{11}-e_{21}), &\theta=\pi/2
         \end{array}
         \right.
\end{align}\label{eq4}
\end{subequations}

 Furthermore, we define the set for \textit{all possible} rotations as:
\begin{subequations}
\begin{align}
\mathcal{A}&\equiv\{ R(\theta,\psi,\phi)| \theta\in[0,\pi),\psi\in[0,\pi),\phi\in[0,4\pi)  \}.
\end{align}
For both $\Theta\in(0,\pi]$ for Type I and $\Theta\in(0,\pi/2]$ for Type II qubits, we define several sets of rotation with $\phi\in[0,4\pi)$:
\begin{align}
\mathcal{G}_{p}&\equiv\{R(\theta,0, \phi)|\theta\in[0,\pi), \phi\in(0,4\pi) \}, \\
\mathcal{G}_{\xi}&\equiv\{ R(\theta,0,\phi)| \theta\in[0,\Theta], \phi\in[0,4\pi)  \}, \\
\mathcal{G}_{z}&\equiv\{ R(\hat{z},\phi)|\phi\in[0,4\pi)  \}, \quad\text{(all $z$ rotations)} \\
\mathcal{G}_{m}&\equiv\{ R(\hat{n},\phi)|\hat{n}=(\sin\Theta,0,\cos\Theta), \phi\in[0,4\pi) \},\\
\mathcal{G}_{b}&\equiv\mathcal{G}_{z}\cup\mathcal{G}_{m},
\end{align}
and rotation with $\phi\in(0,2\pi)$:
\begin{align}
\mathcal{S}_{p}&\equiv\{R(\theta,0, \phi)|\theta\in[0,\pi), \phi\in(0,2\pi) \},\\
\mathcal{S}_{p'}&\equiv\{R(\theta,0, \phi)|\theta\in(0,\pi/2), \phi\in(0,2\pi) \},\\
\mathcal{S}_{\xi}&\equiv\{R(\theta,0, \phi)|\theta\in[0,\Theta], \phi\in(0,2\pi) \},\\
\mathcal{S}_{m}&\equiv\{R(\hat{n}, \phi)|\hat{n}=(\sin\Theta,0,\cos\Theta), \phi\in(0,2\pi) \},\\
\mathcal{S}_{z}&\equiv\{R(\hat{n}, \phi)|\hat{n}=\hat{z}, \phi\in(0,2\pi) \}.
\end{align}
\end{subequations}
Furthermore, given two sets $\mathcal{G}_1, \mathcal{G}_2$, we define the product of them as:
	\begin{equation}
	\mathcal{G}_{1}\mathcal{G}_2\equiv\{ R=R_{1}R_2|R_{1}\in\mathcal{G}_1,R_{2}\in\mathcal{G}_2  \},
	\end{equation}
and for a set $\mathcal{G}$, we define the $p$-power of it as
	\begin{equation}
	\mathcal{G}^p\equiv\{ R=\prod_{i=1}^{p}R_{i}|R_{i}\in\mathcal{G}\}.
	\end{equation}

\section{ Axes restricted in a range}\label{app:B}
Here, we are given  axes that are allowed to vary in a range: $\hat{n}_i=(\sin\theta,0,\cos\theta)$, where $\theta\in[0,\Theta]$, with $\Theta\in(0,\pi]$. We will give the condition for decompositions to exist, and discuss how these decompositions can be constructed or reduced to a Type II qubit case.

\subsection{Lemmas}
We first provide several useful lemmas. To begin with, we show that arbitrary rotations can be decomposed into a $z$-rotation and another rotation with axis in the $x$-$z$ plane. 
\begin{lemma}\label{le:two}
Given any $U(\theta,\psi,\phi)\in\mathcal{A}$, there exist certain $R^z_{1,2}=R(\hat{z},\phi_{1,2})\in\mathcal{G}_z, R_-=R(\theta_-,0,\phi_-)\in\mathcal{G}_{p}, R_+=R(\theta_+,0,\phi_+)\in\mathcal{G}_{p}$, such that 

\begin{subequations}
\begin{equation}
U(\theta,\psi,\phi)=R_1^zR_-,
\end{equation}
and 
\begin{equation}
U(\theta,\psi,\phi)=R_+R_2^z.
\end{equation}\label{decomp2_0}
\end{subequations}

\end{lemma}

\begin{proof}

 \textbf{Case I}: $\phi\in\{0,2\pi\}$ or $\theta=0$
        
       Eq.~\eqref{decomp2_0} can be satisfied by taking  $\phi_{1,2}=\phi$ and $\phi_\pm=0$. 

 \textbf{Case II}: $\phi\notin\{0,2\pi\}$ and $\theta\neq0$

It can be verified that Eq.~\eqref{decomp2_0} can be \textbf{uniquely} constructed as
 \begin{widetext}
\begin{subequations}\label{tc}
\begin{align}
\theta_{\pm}&=\rm{arccot} \left(\frac{\pm\sin\psi\cos\frac{\phi}{2} + \cos\psi\sin\frac{\phi}{2}\cos\theta }{\sin\frac{\phi}{2}\sin\theta} \right),\label{tca}\\
\phi_\pm&=2\pi+\left[ 2\arccos \left( \cos\frac{\phi}{2}\cos\psi\mp\sin\frac{\phi}{2}\sin\psi\cos\theta\right)  - 2\pi \right]\text{sgn}\left(\sin\frac{\phi}{2} \right),\label{tcb}\\
\phi_1&=2\psi, \label{tcc}\\
\phi_2&=-2\psi~\rm{mod}~4\pi.
\end{align}
\end{subequations}
\end{widetext}
\end{proof}~\\
         
In the following, we discuss the decomposition of the product of two rotations in $\mathcal{S}_\xi$.
	\begin{lemma}\label{le:unq}
	given $U_1=U(\theta_1,0,\phi_1)\in \mathcal{S}_{\xi}$, $U_2=U(\theta_2,0,\phi_2)\in \mathcal{S}_{\xi}$ with $\theta_1<\theta_2$, and $\theta_3\in[0, \theta_1]$, there exist \textbf{unique} value of $\phi_3$, and \textbf{unique} $R(\theta_4,0,\phi_4)\in \mathcal{S}_{p'}$, such that
	\begin{equation}
	U_1 U_2 =  R(\theta_3,0,\phi_3) R(\theta_4,0,\phi_4),
	 \end{equation}
	and $\theta_4\neq\theta_3$.

	\end{lemma}

\begin{proof}	
        
	\textit{Existence of $\phi_3$ and $R(\theta_4,0,\phi_4)$}: 
	
	Let $\tilde{\theta}_{1,2}=\theta_{1,2}-\theta_3\in [0,\Theta]$, and define
		\begin{equation}
	 U(\tilde{\theta}_1,0,\phi_1)U(\tilde{\theta}_2,0,\phi_2)=U(\tilde{\theta},\tilde{\psi},\tilde{\phi})=\left[\begin{array}{cc}
	a_{11} & a_{12}\\
	a_{21} &a_{22}
	\end{array}\right]\label{tilde_def}.
	 \end{equation}
         According to Lemma \ref{le:two}, there exist certain $R(0,0,\phi_3)\in\mathcal{G}_z$, $R(\tilde{\theta}_4,0,\phi_4)\in \mathcal{G}_p$, such that 
         
         \begin{equation}
         U(\tilde{\theta},\tilde{\psi},\tilde{\phi})= R(0,0,\phi_3) R(\tilde{\theta}_4,0,\phi_4).\label{tilde_rel}
         \end{equation}
         Since 
         \begin{equation}
         \left|\text{Re}[a_{12}] \right|= \left |\sin\frac{\phi_1}{2}\sin\frac{\phi_2}{2}\sin \tilde{\theta}_1-\tilde{\theta}_2 \right|>0, \label{re12}
         \end{equation}
         we have $\tilde{\phi}\in(0,2\pi)$, and $\tilde{\theta}\neq0$. And combining Eq.~\eqref{tilde_def},~Eq.~\eqref{tilde_rel}, and Eq.~\eqref{tc}, after some calculation, one can verify that 
         \begin{subequations}
         \begin{align}
         \tilde{\theta}_4\in(0,\pi/2),\\
         \phi_3\in(0,2\pi),\\
         \phi_4\in(0,2\pi).
         \end{align}
         \end{subequations}
        
         \begin{widetext}
         Then, we apply a transformation on Eq.~\eqref{tilde_rel}   $S\rightarrow R(\hat{y},\theta_3)R \tilde{S} R(\hat{y},-\theta_3)$, which then becomes
         \begin{align}
         U(\tilde{\theta}_1+\theta_3,0,\phi_1)U(\tilde{\theta}_2+\theta_3,0,\phi_2)&= R(\theta_3,0,\phi_3) R(\tilde{\theta}_4+\theta_3,0,\phi_4)\notag\\
         U(\theta_1,0,\phi_1)U(\theta_2,0,\phi_2)&= R(\theta_3,0,\phi_3) R(\theta_4,0,\phi_4),
         \end{align}        
         where $\theta_4=\tilde{\theta}_4+\theta_3\in[0,\pi)$. So obviously, $R(\theta_3,0,\phi_3)\in \mathcal{S}_{xz}, R(\theta_4,0,\phi_4)\in \mathcal{S}_{p'}$.  
         $\\$
         
         \textit{$\theta_4\neq\theta_3$}:
         
          We denote $R(\tilde{\theta}_4,0,\phi_4)=\left[\begin{array}{cc}
	b_{11} & b_{12}\\
	b_{21} &b_{22}
	\end{array}\right]$. According to Eq.~\eqref{tilde_def}, and  Eq.~\eqref{tilde_rel},  we have
         \begin{equation}
         \left| \sin\frac{\phi_4}{2}\sin\tilde{\theta}_4 \right|=|b_{12}|=|a_{12}|\geqslant \left|\text{Re}[a_{12}] \right|>0.
         \end{equation}
          Therefore, we have $\tilde{\theta}_4\neq 0$, which means $\theta_4\neq\theta_3$.

	\textit{Uniqueness}:
	Suppose 
	\begin{align}
	 U(\theta_1,0,\phi_1) U(\theta_2,0,\phi_2) = R(\theta_3,0,\phi_3) R(\theta_4,0,\phi_4) = R(\theta_3,0,\phi'_3) R(\theta'_4,0,\phi'_4),\label{eq:5}
	 \end{align}
        for some $\phi_{3,4}\in(0,2\pi), \phi'_{3,4}\in(0,2\pi)$, and $\theta_4\neq\theta_{3}, \theta'_4\neq\theta_{3}$. We can denote
	\begin{equation}
	\left[\begin{array}{cc}
	e_{11} &e_{12}\\
	e_{21} &e_{22}
	\end{array}\right]=R(\theta_3,0,\phi_3-\phi'_3)= R(\theta'_4,0,\phi'_4)R(\theta_4,0,-\phi_4).\label{eq:6}
	\end{equation}
	One can find that   
	\begin{equation}
	\text{Re}[e_{12}]=\text{Re}\left[-i\left(\sin\frac{\phi_3-\phi'_3}{2}\sin\theta_3\right)e^{-i0}\right]=0=\sin\frac{\phi'_4}{2}\sin\frac{\phi_4}{2}\sin(\theta'_4-\theta_4).
	\end{equation}
	 	 \end{widetext}
	 And since $\phi_4,\phi_4' \neq 0$, we have 
	\begin{equation}
	\theta'_4=\theta_4.
	\end{equation}
	 So Eq.~\eqref{eq:6} becomes
	\begin{equation}
	R(\theta_3,0,\phi_3-\phi'_3)= R(\theta_4,0,\phi'_4-\phi_4).
	\end{equation}
	Since $\theta_4\neq\theta_3$, and $\phi_{3,4}\in(0,2\pi), \phi'_{3,4}\in(0,2\pi)$, we have 
	\begin{subequations}
	\begin{align}
	\phi_3=\phi'_3,\\
	\phi_4=\phi'_4.
	\end{align}
	\end{subequations}
	 Therefore, the values of $\theta_3,\phi_3$, $\phi_4$ are unique.
	\end{proof}
	
	~\\

\begin{lemma}\label{le:bound1}

Given $U_{1}=U(\theta_1,0,\phi_1)\in \mathcal{S}_{\xi}$, $U_{2}=U(\theta_2,0,\phi_2)\in \mathcal{S}_{\xi}$, with $\theta_1\neq\theta_2$, there exist certain $R^m\in\mathcal{R}_m$, $R^z\in\mathcal{R}_z$ and $R^\xi \in \mathcal{S}_{\xi}$, such that 

(i) if $\theta_1<\theta_2$
  \begin{equation}
  U_1U_2=R^zR^\xi,
  \end{equation}
  
(ii) if $\theta_1<\theta_2$
\begin{equation}
U_1U_2=R^mR^{\xi}.
\end{equation}

\end{lemma}

~\\	
	\begin{proof}~\\~\\
	\textbf{Case I $\theta_1 < \theta_2$ :}
	
	According to Lemma \ref{le:unq}, we can define the following implicit functions $\phi_{3}(x),\phi_{4}(x)$, $y(x)$ that satisfy 
\begin{equation}
R\left(x,0,\phi_3(x)\right)R\left(y(x),0,\phi_4(x)\right)=U_1U_2=\text{Const},
 \label{Const}
\end{equation}	
where $x\in[0,\theta_1]$, $\phi_{3,4}(x)\in(0,2\pi)$ and $y(x)\in(0,\pi)$. From Lemma \ref{le:unq}, the above implicit functions have the following properties:~\\ ~\\
(1)  $y(x)$, $\phi_{3,4}(x)$ are single-value functions (uniqueness); ~\\
(2) $y(x)-x\neq0$;~\\
(3) $y(\theta_1)=\theta_2$.~\\
	To prove case I of Lemma \ref{le:bound1}, we only need to show that $y(\Theta)\in[0,\Theta]$.
        We first evaluate the continuity and monotonicity of $y(x)$. For an independent value $x_0\in[0,\theta_1]$, we always have  
		 \begin{widetext}

	\begin{equation}
	R(x_0,0,\phi_3(x_0))R(y_0,0,\phi_4(x_0))=R(x,0,\phi_3(x))R(y(x),0,\phi_4(x)),
        \end{equation}
which can be rewritten as:
	\begin{eqnarray}  
	R(x,0,-\phi_3(x))  R(x_0,0,h(x_0))   =R(y,0,\phi_4(x))   R(y(x_0),0,-\phi_4(x_0)) 
	=\left[\begin{array}{cc}
	c_{11} & c_{12}\\
	c_{21} &c_{22}
	\end{array}\right].
\end{eqnarray}
	We note that 
	\begin{equation}
	\text{Re}[c_{21}]=\sin\frac{\phi_3(x_0)}{2}\sin\frac{\phi_3(x)}{2}\sin( x_0-x)=\sin\frac{\phi_4(x_0)}{2} \sin\frac{\phi_4(x)}{2}   \sin(y(x_0)-y(x)). \label{hh}
	\end{equation}
	               	 \end{widetext}

	Since  $\phi_{3,4}\in (0,2\pi)$, when $x\neq x_0$, we have
	\begin{equation}
         0<\frac{\sin (y-y_0)}{\sin(x-x_0)}<\infty. \label{eq:cm1}
	\end{equation}
        	And since  $y\in(0,2\pi)$, $y_0\in(0,2\pi)$, we have 
	\begin{equation}
        0< \lim_{x\rightarrow x_0}\frac{y-y_0}{x-x_0}<\infty.\label{eq:cm2}
	\end{equation}	
	Therefore, $y(x)$ and $y(x)-x$ are continuous. Since $y(x)-x\neq 0$ [property (2)], and $y(\theta_1)-\theta_1<0$ [property (3)], according to the intermediate value theorem of continuous function, we have $y(0)-0<0$. 
	
	Moreover, from Eq.~\eqref{eq:cm2} , we know that $y(x_1)\leqslant y(x_2)$ if $x_1\leqslant x_2$. Since $y(\theta_1)<\Theta$, and $0\leqslant\theta_1$, we have $y(0)\leqslant y(\theta_1)<\Theta$. Therefore, $y(\Theta)\in[0,\Theta]$, and (i) of Lemma \ref{le:bound1} hold true.
   
         \textbf{Case II}: $\theta_2 > \theta_1$
         
        In this case, we first let $\tilde{U}_{1,2}=R(\Theta/2,0,\pi)U_{1,2}R(\Theta/2,0,-\pi)$, one can verify that $\tilde{U}_{1,2}=U(\Theta-\theta_{1,2},0,\phi_{1,2})\in\mathcal{S}_\xi$. Since $\Theta-\theta_1<\Theta-\theta_2$, according to case I of Lemma \ref{le:bound1}, there exist certain $\tilde{R}_3=R(\Theta,0,\phi_3)\in\mathcal{S}_m$, $\tilde{R}_4=R(\tilde{\theta}_4,0,\phi_4)\in\mathcal{S}_\xi$, such that 
        \begin{equation}
        \tilde{U}_1\tilde{U}_2=\tilde{R}_3\tilde{R}_4,
        \end{equation}
        which is equivalent to 
        \begin{align}
        U_1U_2&=R(0,0,\phi_3)R(\Theta-\tilde{\theta}_4,0,\phi_4). 
        \end{align}        
        Obviously, $R(\Theta,0,\phi_3)\in\mathcal{S}_m$, $R(\Theta-\tilde{\theta}_4,0,\phi_4)\in\mathcal{S}_\xi$. So (ii) of Lemma \ref{le:bound1} also holds true.

\end{proof}~\\

In the following, we generalize the above result to a larger sets of rotations [those with $\phi\in[0,4\pi)$].
   \begin{lemma}\label{le:bound2}
        
 Given $U_{1}^{\xi}=U(\theta_{1},0,\phi_{1})\in\mathcal{G}_{\xi}, U_{2}^{\xi}=U(\theta_{2},0,\phi_{2})\in\mathcal{G}_{\xi}$, there exist certain rotations $R^z_{1,2}\in\mathcal{G}_z$, $R^{\xi}_{1,2}\in\mathcal{G}_\xi$, and $R^m_{1,2}\in\mathcal{G}_m$, such that either
 \begin{equation}
     U_1^{\xi}U_2^{\xi}=R^{m}_1R^{\xi}_2=R^{\xi}_1R^z_2,
     \end{equation}
     or
             \begin{equation}
     U_1^{\xi}U_2^{\xi}=R^z_1R^{\xi}_2=R^{\xi}_1R^m_2.
     \end{equation}
        \end{lemma}
        
        \begin{proof}

        We classify the domain of $\theta_{1,2}$ and $\phi_{1,2}$ into four cases:

	\textbf{Case I} $\phi_1\in\{0, 2\pi\}$ or $\phi_2\in\{0, 2\pi\}$
	
         We can take $R_{1,2}^{z}=R_{1,2}^{m}=R_{1,2}^{\xi}=R(0,0,0)$, or $R_{1,2}^{z}=R_{1,2}^{m}=R_{1,2}^{\xi}=R(0,0,2\pi)$. 
	
	\textbf{Case II} $\phi_{1,2}\in(0,2\pi)\cup(2\pi,4\pi)$, and $\theta_1=\theta_2=\theta$.
	
	In this case, we have $U_1U_2=R(\theta,\phi_1+\phi_2)$, so
	\begin{align}
	U_1U_2=R(\hat{z},0) R(\theta,0,\phi_1+\phi_2)=R(\theta,0,\phi_1+\phi_2)R(\hat{m},0),
	\end{align}
	and
	\begin{align}
	U_1U_2=R(\hat{z},0) R(\theta,0,\phi)=R(\theta,0,\phi)R(\hat{m},0),
	\end{align}
	where $\phi=(\phi_1+\phi_2)~\rm{mod}~4\pi $. Since $R(\theta,0,\phi)\in\mathcal{G}_{\xi}$, $R(\hat{z},\phi)\in\mathcal{G}_{z}$, $R(\hat{m},\phi)\in\mathcal{G}_{m}$, Lemma \ref{le:bound2} hold true in this case.
	
	\textbf{Case III} $\phi_{1,2}\in(0,2\pi)\cup(2\pi,4\pi)$, $\theta_1\neq\theta_2$ and $0\leqslant\theta_{2}<\theta_{1}\leqslant\Theta$.
       
       Let $\phi_{1',2'}=\phi_{1,2}~\rm{mod}~2\pi$, and $U_{1',2'}=R(\theta_{1,2},0,\phi_{1',2'})$, we have 
       \begin{subequations}\label{rhs}
       \begin{align} 
       U_1U_2&=\pm U_{1'}U_{2'},\label{rhs1}\\
       U_1U_2&=\pm \left[ (-U_{2'})^\dag (-U_{1'})^\dag \right]^\dag.  \label{rhs2}
       \end{align}
       \end{subequations}
       Obviously, $U_{1',2'}\in\mathcal{S}_{\xi}$ and $(-U_{1',2'})^\dag\in\mathcal{S}_{\xi}$, and since $\theta_2<\theta_1$, we can apply 
        Lemma \ref{le:bound1}(i) to the r.h.s. of Eq.~\eqref{rhs1}. In other words, there exist certain rotations 
       \begin{subequations}
       \begin{align}
       R^{m}_{1'}&\in\mathcal{G}_m, \\
       R^{\xi}_{2'}&\in\mathcal{S}_\xi\subset\mathcal{G}_\xi,
       \end{align}\label{in1}
       \end{subequations}
       such that 
       \begin{equation}
       U_{1'}U_{2'}=R^m_{1'}R^{\xi}_{2'}. \label{urp1}
       \end{equation}
       Similarly, we can apply 
        Lemma \ref{le:bound1}(ii) to the r.h.s. of Eq.~\eqref{rhs2}. So there exist certain rotations 
       \begin{subequations}
       \begin{align}
       R^{\xi}_{1'}&\in\mathcal{S}_\xi\subset\mathcal{G}_\xi,\\
       R^{z}_{2'}&\in\mathcal{G}_z,  
       \end{align}\label{in2}
       \end{subequations}
       such that
       \begin{equation}
       (-U_{2'})^\dag(-U_{1'})^\dag=R^z_{2'}R^{\xi}_{1'},
       \end{equation}
       which also leads to
       \begin{equation}
       \left[(-U_{2'})^\dag(-U_{1'})^\dag\right]^\dag=(R^{\xi}_{1'})^{\dag} (R^z_{2'})^{\dag}. \label{urp2}
       \end{equation}       
       Combining Eq.~\eqref{rhs},~\eqref{urp1},~\eqref{urp2}, we have
       \begin{equation}
       U_1U_2=\pm R^m_{1'}R^{\xi}_{2'}=\pm (R^{\xi}_{1'})^{\dag} (R^z_{2'})^{\dag}.\label{eq12}
       \end{equation}
       Since 
       \begin{subequations}
       \begin{align}
       \pm R^m_{1'}\in \mathcal{G}_m,\\
       \pm R^{\xi}_{2'}\in \mathcal{G}_\xi,\\
       \pm (R^{\xi}_{1'})^{\dag}\in \mathcal{G}_\xi,\\
       \pm (R^z_{2'})^{\dag}\in\mathcal{G}_z,
       \end{align}
       \end{subequations}
       Lemma \ref{le:bound2} hold true in this case.

        \textbf{Case IV} $\phi_{1,2}\in(0,2\pi)\cup(2\pi,4\pi)$, $\theta_1\neq\theta_2$ and $0\leqslant\theta_{1}<\theta_{2}\leqslant\Theta$.

        The prove of this case follows the same approach in case III.
      \end{proof}  ~\\
      
      Then, we have the following corollary:
     
      $\\$
      \begin{corollary} \label{co:zm}

      \begin{equation}
      \mathcal{G}_{\xi}\mathcal{G}_\xi=\mathcal{G}_b\mathcal{G}_\xi=\mathcal{G}_\xi\mathcal{G}_b. 
      \end{equation}
      \end{corollary}
      
      \begin{proof}
      
      From Lemma \ref{le:bound2}, we know that $ \mathcal{G}_{\xi}\mathcal{G}_{\xi}\subset \left(\mathcal{G}_{z}\mathcal{G}_{\xi}\cap\mathcal{G}_{\xi}\mathcal{G}_{m} \right) \cup \left(\mathcal{G}_{m}\mathcal{G}_{\xi}\cap\mathcal{G}_{\xi}\mathcal{G}_{z} \right)$. Since $\mathcal{G}_{b}\subset\mathcal{G}_{\xi}$ and $\mathcal{G}_{b}=\mathcal{G}_{z}\cup\mathcal{G}_{m}$, we have $\mathcal{G}_{\xi}\mathcal{G}_\xi=\mathcal{G}_b\mathcal{G}_\xi=\mathcal{G}_\xi\mathcal{G}_b$.
      \end{proof}

        \subsection{$p\geqslant3$ piece decomposition} \label{sec:reduce}
        
       In the following, we will provide the proof for (i) of Theorem 1, which is equivalent to the following
\begin{theorem}\label{th:s} 

For $p\in\mathbb{N}^{*}$, $p\geqslant 3$, and $\Theta\in(0,\frac{\pi}{2})$
	\begin{equation}
	\mathcal{G}_{\xi}^p=\mathcal{G}_{b}^p. \label{theta}
	\end{equation}
\end{theorem}

	\begin{proof}

	 From Corollary \ref{co:zm}, we know that  
	 \begin{equation}
	 \mathcal{G}_{\xi}\mathcal{G}_{\xi} \subset \mathcal{G}_{b}\mathcal{G}_{\xi}\cap \mathcal{G}_{\xi}\mathcal{G}_{b}. \label{xi}
	 \end{equation}
	 
	 One can verify from Theorem \ref{th:odd} (refer to section \ref{sec:fix}) that when $\Theta\in(0,\frac{\pi}{2})$, all rotations in $\mathcal{G}_{\xi}$ can be decomposed into three pieces both in the form of $z$-$m$-$z$ and $m$-$z$-$m$. In orther words, we have
	\begin{equation}
	\mathcal{G}_{\xi}\subset \mathcal{G}_{z}\mathcal{G}_{m}\mathcal{G}_{z}\cap\mathcal{G}_{m}\mathcal{G}_{z}\mathcal{G}_{m}, \label{set_zmz}
	\end{equation}
	which also gives 
	\begin{subequations}
	\begin{align}
	\mathcal{G}_b\mathcal{G}_{\xi}\subset \mathcal{G}_{z}\mathcal{G}_{m}\mathcal{G}_{z}\cap\mathcal{G}_{m}\mathcal{G}_{z}\mathcal{G}_{m}, \\
	\mathcal{G}_{\xi}\mathcal{G}_b\subset \mathcal{G}_{z}\mathcal{G}_{m}\mathcal{G}_{z}\cap\mathcal{G}_{m}\mathcal{G}_{z}\mathcal{G}_{m}.
	\end{align}\label{bxizmz}	
	\end{subequations}
	Combining Eq.~\eqref{xi} and Eq.~\eqref{bxizmz}, we have
	\begin{equation}
	\mathcal{G}_{\xi}\mathcal{G}_{\xi}\subset \mathcal{G}_{z}\mathcal{G}_{m}\mathcal{G}_{z}\cap\mathcal{G}_{m}\mathcal{G}_{z}\mathcal{G}_{m}. \label{set_zmz_2}
	\end{equation}
        According to Eq.~\eqref{xi} we have
	\begin{equation}
	(\mathcal{G}_{\xi})^p\subset \mathcal{G}_{b}(\mathcal{G}_{\xi})^{p-1}\subset \mathcal{G}_{b}\mathcal{G}_{b}(\mathcal{G}_{\xi})^{p-2}\subset \cdots \subset \left(\mathcal{G}_{b} \right)^{p-3}\mathcal{G}_{b}\mathcal{G}_{\xi} \mathcal{G}_{\xi} \label{S}.
	\end{equation}
	Combining Eq.~\eqref{set_zmz_2},~\eqref{S}, and note that $\mathcal{G}_{b}=\mathcal{G}_{z}\cup\mathcal{G}_{m}$, we have 
	\begin{equation}
	\mathcal{G}_{\xi}^p\subset\mathcal{G}_{b}^p \label{S_2},
	\end{equation}
        and since 
	\begin{equation}
	\mathcal{G}_{b}^p\subset\mathcal{G}_{\xi}^p,
	\end{equation}        
	we finally get 
	\begin{equation}
	\mathcal{G}_{\xi}^p=\mathcal{G}_{b}^p.
	\end{equation}
	\end{proof}

        \subsection{$p=2$ decomposition}
        \begin{theorem} \label{th:two} ~\\~\\
        Given $U(\theta,\psi,\phi)\in\mathcal{A}$, 
        \begin{equation}
        U(\theta,\psi,\phi)\in\mathcal{G}_{\xi}\mathcal{G}_{\xi} \label{two_piece}
        \end{equation}
         if and only if one of the following condition is satisfied:\\
        (i) $\phi\in\{0,2\pi\}$,\\
        (ii) $\theta=0$,\\
        (iii)
\begin{equation}
\frac{\pm\sin\psi\cos\frac{\phi}{2} + \cos\psi\sin\frac{\phi}{2}\cos\theta }{\sin\frac{\phi}{2}\sin\theta}\geqslant \cot\Theta \label{two_cr}
\end{equation}
is satisfied for either sign of `$\pm$', or\\
(iv) $\Theta=\pi$.
        \end{theorem}
        
        \begin{proof}

        \textbf{Case I} (i) $\phi\in\{0,2\pi\}$ or (ii) $\theta=0$:
        
        Eq.~\eqref{two_piece} can always be constructed by taking $R_1=R(\hat{z},\phi), R_2=R(\hat{z},0)$.~\\
               
        \textbf{Case II} $\phi\in(0,2\pi)\cup(2\pi,4\pi)$, and $\theta\in(0,\pi)$: 
        
        In such case, we should show that the existence of decomposition as Eq.~\eqref{two_piece} is equivalent to (iii) or (iv).
        
        According to Lemma \ref{le:two}, $U(\theta,\psi,\phi)$ can always be written as
        \begin{subequations}
        \begin{align}
        U(\theta,\psi,\phi)&=R(\hat{z},\phi_1)R(\theta_-,0,\phi_-), \label{two_1} \\ 
        U(\theta,\psi,\phi)&=R(\theta_+,0,\phi_+)R(\hat{z},\phi_2), \label{two_2}
        \end{align}        
        \end{subequations}
        for certain values of $\phi_{1,2}\in[0,4\pi)$, $\phi_{\pm}\in[0,4\pi)$, and
        \begin{subequations} \label{thetapm}
        \begin{eqnarray}
        \cot\theta_-= \frac{-\sin\psi\cos\frac{\phi}{2} + \cos\psi\sin\frac{\phi}{2}\cos\theta }{\sin\frac{\phi}{2}\sin\theta},\\ 
        \cot\theta_+= \frac{\sin\psi\cos\frac{\phi}{2} + \cos\psi\sin\frac{\phi}{2}\cos\theta }{\sin\frac{\phi}{2}\sin\theta}.
        \end{eqnarray}
        \end{subequations}
        We notice that in case II, the values of $\theta_{\pm}$ are unique. We introduce the following statements~\\~\\
        $a.$ $U(\theta,\psi,\phi)$ satisfies (iii) or (iv);~\\
        $b.$ $R(\theta_-,0,\phi_-)\in \mathcal{G}_{\xi},$ or $R(\theta_+,0,\phi_+)\in \mathcal{G}_{\xi}$; ~\\
        $c.$ $U(\theta,\psi,\phi)\in\mathcal{G}_{z}\mathcal{G}_{\xi}\cup\mathcal{G}_{\xi}\mathcal{G}_{z}$;~\\
        $d.$ $U(\theta,\psi,\phi)\in\mathcal{G}_\xi\mathcal{G}_\xi$.~\\
        
        From Eq.~\eqref{thetapm}, one can verify that $a\Leftrightarrow b$, and since the value of $\theta_{\pm}$ are unique, we have $b\Leftrightarrow c$.
        From Lemma \ref{le:bound2}, we know that $\mathcal{G}_{\xi}\mathcal{G}_{\xi}=\mathcal{G}_{z}\mathcal{G}_{\xi}\cup\mathcal{G}_{\xi}\mathcal{G}_{z} $, so $c\Leftrightarrow d$. Therefore, $a\Leftrightarrow d$, and Theorem.~\ref{th:two} holds.
        \end{proof}

\section{ Decomposition with two fixed axes}\label{sec:fix}
Here, we are given two fixed axes $\hat{z}=(0,0,0)$ and $\hat{m}=(\sin\Theta,0,\cos\Theta)$, and the angle between them is restricted to $\Theta\in(0,\frac{\pi}{2}]$. We are going to prove the criteria for fixed-axes decomposition  [Eq.~(8) and Eq.~(9)], and provide methods to construct the decomposition sequences. 

\subsection{Odd-piece decomposition}
\subsubsection{Criterion for odd-piece decomposition}
For odd-piece decomposition, i.e. $p=2l-1$ with $l\in \mathbb{Z}^{+}$, Eq.~\eqref{general} is equivalent to 
\begin{widetext}

\begin{equation}
U \left( \theta,\psi ,\phi \right) = R\left( \hat{z}, \beta_0\right) R\left( \hat{m}, \gamma_1 \right) R \left(\hat{z},\beta_1 \right) \ldots R\left(\hat{m}, \gamma_{l-1} \right)R(\hat{z},\beta_{l-1}),\label{odd1}
\end{equation}
or 
\begin{equation}
	 U \left( \theta,\psi,\phi \right) = R\left( \hat{m}, \beta_0\right) R\left( \hat{z}, \gamma_1 \right) R \left(\hat{m},\beta_1 \right) \ldots R\left(\hat{z}, \gamma_{l-1} \right)R(\hat{m},\beta_{l-1}),\label{odd2}
\end{equation}
where $\beta_i\in[0,4\pi)$, $\gamma_i\in[0,4\pi)$. We define 
\begin{subequations}
\begin{align}
\delta_1(\theta,\phi)&=\arcsin \left(\sin\theta\sin\frac{\phi}{2} \right), \label{delta_ele}\\
\delta_2(\theta,\psi,\phi,\Theta)&=\arcsin \left[ \sin\frac{\phi}{2}\sqrt{(\cos\Theta\cos\psi\sin\theta
-\cos\theta\sin\Theta)^{2}
+ (\sin\theta\sin\psi)^{2}} \right].
\end{align}
\end{subequations}
       	 \end{widetext}

Before giving the proof of theorem, we first provide some useful lemmas. 
\begin{lemma} ($z$-$m$-$z$ decomposition)\label{le:zmz}~\\
Given $U(\theta,\psi,\phi)\in\mathcal{A}$, $\Theta\in(0,\frac{\pi}{2}]$, and $\Psi\in[0,\pi]$, there exist certain values of $\beta'_{0,1}\in[0,4\pi)$, $\gamma'_{1}\in[0,4\pi)$, such that 
\begin{equation}
U(\theta,\psi,\phi)=R(\hat{z},\beta'_0)R(\Theta,\Psi ,\gamma'_1)R(\hat{z},\beta'_1), \label{three}
\end{equation}	
if and only if $|\delta_1(\theta,\phi)|\leqslant \Theta$.
\end{lemma}

\begin{proof}

\textit{Necessity of $|\delta_1(\theta,\phi)|\leqslant \Theta$}:

Let $U(\theta,\psi,\phi)= \left[\begin{array}{cc}
	e_{11} & e_{12}\\
	e_{21} &e_{22}
	\end{array}\right],$ 
according to Eq.~\eqref{SU2} and Eq.~\eqref{delta_ele}, when Eq.~\eqref{three} holds, we have

\begin{equation}
|e_{12}|=\sin|\delta_1(\theta,\phi)|=\left|\sin\Theta\sin\frac{\gamma'_1}{2}\right|.\label{e12}
\end{equation} 
Obviously, for $|\delta_1(\theta,\phi)|>\Theta $, Eq.~\eqref{e12} cannot be satisfied for any $\gamma'_1$, so the decomposition as Eq.~\eqref{three} does not exist. ~\\

\textit{Sufficiency of $|\delta_1(\theta,\phi)|\leqslant \Theta$}:

When $\delta_1(\theta,\phi)\leqslant \Theta $ is satisfied, Eq.~\eqref{three} can be constructed as:
\begin{equation}
\gamma'_1=\pi\pm \left [2\arcsin \left(\frac{\sin\theta\sin\frac{\phi}{2}}{\sin\Theta} \right )-\pi \right], 
\end{equation}
and 
\begin{subequations}
\begin{eqnarray}
\beta'_0=\alpha_3-\alpha_1,\\
\beta'_1=\alpha_4-\alpha_2,
\end{eqnarray}
\end{subequations}
where
$\alpha_1=-\psi-\lambda_1 ,\alpha_2=\psi-\lambda_1,\alpha_3=-\Psi-\lambda_2,\alpha_4=\Psi-\lambda_2,$
and
\begin{subequations}	
\begin{eqnarray}
\lambda_1=\text{Arg} \left( \cos\frac{\phi}{2}+i\sin\frac{\phi}{2}\cos\theta \right),\\
\lambda_2=\text{Arg} \left( \cos\frac{\gamma'_1}{2}+i\sin\frac{\gamma'_1}{2}\cos\Theta \right).
\end{eqnarray}	\label{zmz_ed}
\end{subequations}
\end{proof}~\\

Since $|\delta_1(\theta,\phi)|\leqslant \frac{\pi}{2}$,  a three-piece decomposition for arbitrary rotations always exists when $\Theta=\pi/2$. In particular, we have the following corollary:
\begin{corollary} ($z$-$x$-$z$ decomposition) \label{co:zxz} ~\\
Given $U(\theta,\psi,\phi)\in \mathcal{A}$, it can always be decomposed as
\begin{equation}
U(\theta,\psi,\phi)=R(\hat{z},\beta_0)R(\hat{x},\gamma)R(\hat{z},\beta_1), \label{zxz}
\end{equation} 
where
\begin{subequations}
\begin{align}
\beta_0&=\rm{Arg}\left( \cos\frac{\phi}{2} +i\sin\frac{\phi}{2}\cos\theta \right)+\psi,\\
\beta_1&=\rm{Arg}\left( \cos\frac{\phi}{2} +i\sin\frac{\phi}{2}\cos\theta \right)-\psi,\\
\gamma&=2\arcsin\left(\sin\theta\sin\frac{\phi}{2}   \right).
\end{align}
\end{subequations}
\end{corollary}~\\
We now generalize Lemma \ref{le:zmz} to an arbitrary odd number of pieces.
       	 \begin{widetext}
\begin{lemma}\label{le:l}  ~\\
Given a rotation $U(\theta,\psi,\phi)\in\mathcal{A}$, there exist certain values of $\beta'_i\in[0,4\pi)$, $\gamma'_i\in[0,4\pi)$, and $l\in\mathbb{Z}^{+}$, such that 
	\begin{equation}
	 U \left( \theta,\psi ,\phi \right) =  R\left( \hat{z}, \beta'_0\right) R\left( \Theta,\Psi, \gamma'_1 \right) R \left(\hat{z},\beta'_2 \right) \ldots R\left(\Theta,\Psi, \gamma'_{l-1} \right)R(\hat{z},\beta'_{l-1}),  \label{odd_0}
	\end{equation}
         if and only if 
         \begin{equation}
         |\delta_1(\theta,\phi)|\leqslant (l-1)\Theta.\label{odd_cr}
         \end{equation} 
\end{lemma}
~\\
\end{widetext}
\begin{proof}

\textbf{Case (i)}: $l=1$.

In this case, Eq.~\eqref{odd_0} and Eq.~\eqref{odd_cr} become 
\begin{subequations}
\begin{align}
U \left( \theta,\psi ,\phi \right) &=  R\left( \hat{z}, \beta'_0\right) \label{odd_a},\\
|\delta_1(\theta,\phi)|& = 0. \label{odd_b}
\end{align} 
\end{subequations}
Obviously, both Eq.~\eqref{odd_a} and Eq.~\eqref{odd_b} are equivalent to $\theta=0$ or $\phi\in\{0,2\pi\}$.~\\

\textbf{Case (ii)}: $l>1$.

\textit{Necessity of $|\delta_1(\theta,\phi)|\leqslant (l-1)\Theta$}:~\\
 Let 
\begin{subequations}
\begin{align}
\gamma''_i=\gamma'_i~\rm{mod}~2\pi,\\
\beta''_i=\beta'_i~\rm{mod}~2\pi,
\end{align} 
\end{subequations}
Eq.~\eqref{odd_0} is equivalent to 
		\begin{align}
	 U \left( \theta,\psi ,\phi \right) 
	 &=\pm R\left( \hat{z}, \beta''_0\right) R\left( \Theta,\Psi, \gamma''_1 \right)  \cdots R(\hat{z},\beta''_{l-1}).\label{odd}
	\end{align}
	According to corollary \ref{co:zxz}, one can apply the $z$-$x$-$z$ decomposition on each $R(\Theta,\Psi,\gamma''_{i})$. So if Eq.~\eqref{odd_0} holds, $U(\theta,\psi,\phi)$ can be further rewritten as 
	\begin{align}
	 &U \left(\theta,\psi,\phi \right) \notag\\
	 &= \pm R \left(\hat{z}, \eta_0\right) R\left(\hat{x}, \rho_1 \right) R \left(\hat{z}, \eta_1 \right) \ldots R
	\left(\hat{x}, \rho_{l-1} \right) R \left(\hat{z}, \eta_{l-1} \right),   \label{odd_zxz}
	\end{align}
	for certain values of $\eta_i\in [0,2\pi)$, and $\rho_i= 2\arcsin\left(\sin\Theta\sin\frac{\gamma''_1}{2}\right)$. Since $\Theta\in(0,\frac{\pi}{2}]$, $\gamma''_i\in[0,2\pi)$, we have  
	\begin{eqnarray}
	0 \leqslant \rho_i\leqslant 2\Theta\leqslant \pi. \label{ineq:rho}
	\end{eqnarray}
	We give two statements:
	(a)  $|\delta_1(\theta,\phi)|>(l-1)\Theta$, and	
	(b)  Eq.~\eqref{odd_zxz} holds.
	
	Since Eq.~\eqref{odd_zxz} is equivalent to Eq.~\eqref{odd_0}, to prove the necessity of Lemma \ref{le:l}, we only need to show that (a) and (b) cannot be satisfied at the same time. In the following, we assume that both (a) and (b) are satisfied.

	 We define 
	 \begin{align}
	 B_t &\equiv
	\left[\begin{array}{cc}
	b_{t,11} & b_{t,12}\\
	b_{t,21} & b_{t,22}
	\end{array}\right]\notag\\
	&=R \left(\hat{z}, \eta_0 \right) R \left(\hat{x}, \rho_1 \right) R \left(\hat{z}, \eta_1\right) \ldots R \left(\hat{x}, \rho_t \right) R \left(\hat{z}, \eta_t \right), 
	\end{align}
	where $t \leqslant l-1$. Note that $B_{l-1}=R(\theta,\psi,\phi)$, and
	\begin{equation}
	|b_{l-1,11}|=\cos\delta_1(\theta,\phi). \label{bd}
	\end{equation}	
	Then, the value of $|b_{t,11}|$ will be bounded by induction as follows.
	
	For $t = 1$, Eq.~\eqref{ineq:rho} implies that $|
	b_{1,11} | = \cos \frac{\rho_1}{2} \geqslant \cos \Theta$;	
	for $1<t\leqslant (l-1)$, we suppose $| b_{t-1,11} | \geqslant \cos \left[ (t-1) \Theta \right]$ holds.
	 One can let 
	 \begin{subequations}
	 \begin{eqnarray}
	 b_{t-1,11} =e^{i \varphi_1} \cos \alpha,\\
	 b_{t-1,12} = e^{i \varphi_2} \sin \alpha,
	 \end{eqnarray}
	 \end{subequations}
	  for certain values of $0\leqslant\alpha \leqslant (t-1) \Theta$, and $0\leqslant\varphi_{1,2}<2\pi$.
	Since 
	\begin{equation}
	 \left[\begin{array}{cc}
	b_{t,11} & b_{t,12}\\
	b_{t,21} & b_{t,22}
	\end{array}\right] = \left[\begin{array}{cc}
	b_{t-1,11} & b_{t-1,12}\\
	b_{t-1,21} & b_{t-1,22}
	\end{array}\right] R_{\hat{x}} \left( \rho_{t} \right)
	R_{\hat{z}} \left( \eta_{t} \right), 
	\end{equation}
	we have 
	\begin{equation}
	 b_{t,11} = e^{i \varphi_1} \cos \alpha \cos \frac{\rho_{t}}{2} - i e^{i \varphi_2}
	\sin \alpha \sin \frac{\rho_{t}}{2} . 
	\end{equation}
	
	\begin{widetext}
	Then
	\begin{align}
	| b_{t,11} |^2 = & \cos^2 \alpha \cos^2 \frac{\rho_{t}}{2} + \sin^2 \alpha \sin^2
	\frac{\rho_{t}}{2} + 2 \sin \left( \varphi_1 - \varphi_2 \right) \cos \alpha \cos
	\frac{\rho_{t}}{2} \sin \alpha \sin \frac{\rho_{t}}{2} & \notag\\
	= & \left( \cos \alpha \cos \frac{\rho_{t}}{2} - \sin \alpha \sin \frac{\rho_{t}}{2}
	\right)^2 - 2 \left( 1 - \sin \left( \varphi_1 - \varphi_2 \right) \right)
	\cos \alpha \cos \frac{\rho_{t}}{2} \sin \alpha \sin \frac{\rho_{t}}{2} & \notag\\
	\geqslant & \left( \cos \alpha \cos \frac{\rho_{t}}{2} - \sin \alpha \sin \frac{\rho_{t}}{2} \right)^2 &\notag \\
	= & \cos^2 \left( \alpha + \frac{\rho_{t}}{2} \right) &\notag \\
	\geqslant & \cos^2  t \Theta . & 
	\end{align}
	\end{widetext}

	The last inequality is due to $\alpha \leqslant (t-1)\Theta$, $0\leqslant \rho_t \leqslant 2\Theta$, and $t\Theta\leqslant (l-1)\Theta<\delta_1(\theta,\phi)\leqslant \frac{\pi}{2}$. Therefore, if both (a) and (b) hold true, we have $| b_{t,11} | \geqslant \cos \left( t\Theta \right)$ for $1\leqslant t\leqslant l$, which also gives
	\begin{equation}
	| b_{l-1,11} | \geqslant \cos (l-1)\Theta.\label{bound}
	\end{equation}
Combining Eq.~\eqref{bd}, Eq.~\eqref{bound} and $\delta_1(\theta,\phi)\in[0,\frac{\pi}{2}]$, we have $\delta_1(\theta,\phi) \leqslant (l-1)\Theta$. However, this is contradicted to (a). Therefore (a) and (b) cannot be satisfied at the same time, which finish the proof of necessity.~\\
 	
\textit{Sufficiency of $|\delta_1(\theta,\phi)|\leqslant (l-1)\Theta$}:
	
The sufficiency will be proven constructively. According to Corollary \ref{co:zxz}, $U \left( \theta,\psi,\phi \right)$  can first be decomposed as:
	\begin{align}
        U \left( \theta,\psi,\phi \right)=R \left(\hat{z},\lambda_1+\psi \right)  R \left(\hat{x}, 2\delta_1(\theta,\phi) \right)  R \left(\hat{z}, \lambda_1-\psi \right),
         \end{align}
        where 
        \begin{equation}
        \lambda_1=\text{Arg} \left( \cos\frac{\phi}{2}+i\sin\frac{\phi}{2}\cos\theta \right).
        \end{equation}
     The   $x$ rotation in the middle can be divided into $l-1$ pieces, and we get 
	\begin{align}
        &U \left( \theta,\psi,\phi \right)\notag\\
        &= R \left(\hat{z},\lambda_1+\psi \right) \left[ R \left(\hat{x}, \frac{2\delta_1(\theta,\phi)}{l-1} \right) \right]^{l-1} R \left(\hat{z}, \lambda_1-\psi \right).\label{A6}
        \end{align}

We notice that $R\left(\hat{x}, \frac{2\delta_1(\theta,\phi)}{l-1} \right)=R\left(\frac{\pi}{2},0, \frac{2\delta_1(\theta,\phi)}{l-1} \right)$, and  $\delta_1 \left(\frac{\pi}{2},\frac{2\delta_1(\theta,\phi)}{l-1} \right)=\frac{\delta_1(\theta,\phi)}{l-1}\leqslant\Theta$. According to Lemma \ref{le:zmz}, when $|\delta_1(\theta,\phi)|\leqslant (l-1)\Theta$, we can have the decomposition $R(\hat{x},\frac{2\delta_1(\theta,\phi)}{l-1})=R(\hat{z},-\Psi-\chi)R(\Theta,\Psi,\gamma')R(\hat{z},\Psi-\chi)$, with
\begin{equation}
\gamma'=\pi\pm\left[2\arcsin \left( \frac{\sin\frac{\delta_1(\theta,\phi)}{l-1}}{\sin\Theta} \right) -\pi\right] ,
\end{equation} 
and 
\begin{equation}
\chi=\text{Arg} \left( \cos\frac{\gamma'}{2}+i\sin\frac{\gamma'}{2}\cos\Theta \right).
\end{equation}
So Eq.~\eqref{odd_0} can be constructed by taking.
\begin{equation}  
\beta'_i = \left\{  
             \begin{array}{ll}  
             \lambda_1+\psi-\Psi-\chi,& i=0, \\  
             -2\chi,                               & 0<i<l-1, \\  
             \lambda_1-\psi+\Psi-\chi, &i=l-1,    
             \end{array}  
\right. \label{A10}
\end{equation} 
and 
\begin{equation}
\gamma'_j=\gamma'
\end{equation}
 for $1\leqslant j\leqslant l-1$.
\end{proof}~\\

The following theorem corresponds to the odd-piece decomposition.
\begin{theorem} \label{th:odd}
Given $U(\theta,\psi,\phi)\in\mathcal{A}$, it can be decomposed to $2l-1$ pieces with $l\in \mathbb{Z}^{+}$ 

(i) as Eq.~\eqref{odd1} with certain values of $\beta_i\in[0,4\pi)$, $\gamma_i\in[0,4\pi)$, if and only if 
\begin{equation}
\left|\delta_1(\theta,\phi)\right| \leqslant (l-1)\Theta, \label{crodd_1}
\end{equation}
or (ii) as Eq.~\eqref{odd2} with certain values of $\beta_i\in[0,4\pi)$, $\gamma_i\in[0,4\pi)$, if and only if 
\begin{equation}
\left|\delta_2(\theta,\psi,\phi,\Theta)\right| \leqslant (l-1)\Theta. \label{crodd_2}
\end{equation}

\end{theorem}

\begin{proof}

For case (i), by taking $\Psi=0$ in Lemma \ref{le:l}, one can verify that Theorem \ref{th:odd} holds true.

For case (ii), we first apply the transformation $U\rightarrow R(\Theta/2,0,-\pi) U R(\Theta/2,0,\pi)$ [rotating all axes around $(\sin\Theta/2,0,\cos\Theta/2)$ by angle $\pi$] on Eq.~\eqref{odd2}, and obtain:
\begin{align}
	& U \left( \tilde{\theta},\tilde{\psi},\tilde{\phi} \right) \notag\\
	 &= R\left( \hat{z}, \beta_0\right) R\left( \hat{m},0, \gamma_1 \right)  \ldots R\left(\hat{z},0, \gamma_{l-1} \right)R(\hat{z},0, \beta_{l-1}) \label{mzm_2}
\end{align}
where $U(\tilde{\theta},\tilde{\psi},\tilde{\phi})=R(\Theta/2,0,-\pi)U(\theta,\psi,\phi)R(\Theta/2,0,\pi)$. It is straightforward to verify 
\begin{equation}
\left|\delta_1(\tilde{\theta},\tilde{\phi})\right|=\left|\delta_2(\theta,\psi,\phi,\Theta)\right|.
\end{equation}
Therefore, case (ii) of Theorem \ref{th:odd} also holds true. 
\end{proof}

\subsubsection{Constructing the odd-piece decomposition}\label{sec:odd}
When Eq.~\eqref{crodd_1} is satisfied, Eq.~\eqref{odd1} can be constructed as
\begin{equation}  
\beta_i = \left\{  
             \begin{array}{ll}  
             \lambda_1-\lambda_2+\psi &\qquad i=0, \\  
             -2\lambda_2                               &\qquad 0<i<l-1, \\  
             \lambda_1-\lambda_2-\psi &\qquad i=l-1,   
             \end{array}  
\right. 
\end{equation} 
and 
\begin{equation}
\gamma_j= \begin{array}{ll}   \gamma,&1\leqslant j\leqslant l-1,\end{array}
\end{equation}
where
\begin{subequations}
\begin{align}
\gamma&=\pi\pm\left[2\arcsin \left( \frac{\sin\frac{\delta_1(\theta,\phi)}{l-1}}{\sin\Theta} \right) -\pi\right],\\
\lambda_1&=\text{Arg} \left( \cos\frac{\phi}{2}+i\sin\frac{\phi}{2}\cos\theta \right),\\
\lambda_2&=\text{Arg} \left( \cos\frac{\gamma}{2}+i\sin\frac{\gamma}{2}\cos\Theta \right).
\end{align}
\end{subequations}
Similarly, when Eq.~\eqref{crodd_2} is satisfied, to construct Eq.~\eqref{odd2} we first calculate
\begin{equation}
U(\tilde{\theta},\tilde{\psi},\tilde{\phi})=R(\Theta/2,0,-\pi)R(\theta,\psi,\phi)R(\Theta/2,0,\pi),
\end{equation}
and the corresponding values of $\tilde{\theta},\tilde{\psi}$ and $\tilde{\phi}$. Then, we have 
\begin{equation}  
\beta_i = \left\{  
             \begin{array}{ll}  
             \tilde{\lambda}_1-\tilde{\lambda}_2+\tilde{\psi} &\qquad i=0, \\  
             -2\tilde{\lambda}_2                               &\qquad 0<i<l-1, \\  
             \tilde{\lambda}_1-\tilde{\lambda}_2-\tilde{\psi} &\qquad i=l-1,  
             \end{array}  
\right. 
\end{equation} 
and 
\begin{equation}
\gamma_j= \begin{array}{ll}   \tilde{\gamma},&1\leqslant j\leqslant l-1,\end{array}
\end{equation}
where
\begin{subequations}
\begin{align}
\tilde{\gamma}&=\pi\pm\left[2\arcsin \left( \frac{\sin\frac{\delta_1(\widetilde{\theta},\widetilde{\phi})}{l-1}}{\sin\Theta} \right) -\pi\right],\\
\tilde{\lambda}_1&=\text{Arg} \left( \cos\frac{\tilde{\phi}}{2}+i\sin\frac{\tilde{\phi}}{2}\cos\theta \right),\\
\tilde{\lambda}_2&=\text{Arg} \left( \cos\frac{\tilde{\gamma}}{2}+i\sin\frac{\tilde{\gamma}}{2}\cos\Theta \right).
\end{align}
\end{subequations}
It should be notice that this decomposition method is not the unique one.

 \subsection{Even-piece decomposition}
 
 \subsubsection{Criterion for even-piece decomposition}
 For even-piece decomposition, i.e. $p=2l$ with $l\in \mathbb{Z}^{+}$, Eq.~\eqref{general} is equivalent to 
\begin{equation}
U \left( \theta,\psi ,\phi \right) = R\left( \hat{m}, \beta_1 \right) R \left(\hat{z},\gamma_1 \right) \ldots R\left(\hat{m}, \beta_l \right)R(\hat{z},\gamma_l),\label{even1}
\end{equation}
or 
\begin{equation}
 U \left( \theta,\psi,\phi \right) = R\left( \hat{z}, \beta_1 \right) R \left(\hat{m},\gamma_1 \right) \ldots R\left(\hat{z}, \beta_l \right)R(\hat{m},\gamma_l),\label{even2}
\end{equation}
where $\beta_i=\alpha_{2i-1} \in[0,4\pi)$ and $\gamma_{i}=\alpha_{2i}\in[0,4\pi)$.
We denote
 \begin{align}
    A =& \left( \cos \psi \cos \Theta \sin \theta \sin \frac{\phi}{2} - \sin \Theta \cos \theta \sin \frac{\phi}{2} \right)^2 \notag\\
    &+ \left( \sin \psi \cos \Theta \sin \theta \sin \frac{\phi}{2} - \sin
  \Theta \cos \frac{\phi}{2} \right)^2,\notag\\
  B = & \left( \sin \theta \sin \frac{\phi}{2} \right)^2,\notag\\ 
  C = & \sin \Theta \sin \theta \sin \frac{\phi}{2} (\sin \psi \sin \frac{\phi}{2}\cos \theta -
 \cos \psi  \cos \frac{\phi}{2}),
\end{align}

and 
 \begin{align}
 \Lambda (\theta,\psi,\phi,\Theta)=\arcsin \sqrt{ \frac{A + B}{2} - \sqrt{C^2 + \frac{(B - A)^2}{4}} }.
 \end{align}
 \begin{theorem}\label{th:even}
 
 Given $R(\theta,\psi,\phi)\in\mathcal{A}$, it can be decomposed to $2l$ pieces with $l\in \mathbb{Z}^{+}$
 
 (i) as  Eq.~\eqref{even1} with certain values of $\beta_i\in[0,4\pi)$, $\gamma_i\in[0,4\pi)$, if and only if 
 \begin{equation}
 \Lambda(\theta,\psi,\phi,\Theta)\leqslant (l-1)\Theta, \label{even_cr1}
 \end{equation}
or (ii) as Eq.~\eqref{even2} with certain values of $\beta_i\in[0,4\pi)$, $\gamma_i\in[0,4\pi)$, if and only if 
 \begin{equation}
 \Lambda(\theta,\psi,-\phi,\Theta)\leqslant (l-1)\Theta. \label{even_cr2}
 \end{equation}
 \end{theorem}

\begin{proof}

\textbf{Case (i)}:

We first define
 \begin{equation}
 \left[\begin{array}{cc}
	e_{11}(\beta_1) & e_{12}(\beta_1)\\
	e_{21}(\beta_1) & e_{22}(\beta_1)
	\end{array}\right] \equiv R\left( \hat{m}, -\beta_1 \right) U \left( \theta,\psi,\phi \right).  \label{mr}
 \end{equation}
 According to Eq.~\eqref{delta_ele} and Theorem \ref{th:odd}, the existence of Eq.~\eqref{even1} is equivalent to the existence of $\beta_1\in[0,4\pi)$, such that 
 \begin{equation}
  \arcsin|e_{12}(\beta_1)|\leqslant(l-1) \Theta. \label{arce12}
  \end{equation}
   It can be calculated from Eq.~\eqref{mr} that 
 \begin{align}
  e_{12}(\beta_1)= &e^{- i \psi} \left( - i \cos \frac{\beta_1}{2} + \cos \Theta \sin
  \frac{\beta_1}{2} \right) \sin \theta \sin \frac{\phi}{2} \notag\\
  &+ i \sin
  \frac{\beta_1}{2} \sin \Theta \left( \cos \frac{\phi}{2} + i \cos
  \theta \sin \frac{\phi}{2} \right). 
 \end{align} 
 After some further calculation, one can obtain that 
 \begin{align}
  |e_{12}(\beta_1)|^2&= A\sin^2\frac{\beta_1}{2}+B\cos^2\frac{\beta_1}{2}+C\sin\beta_1 \notag\\
  &=  \frac{(B - A)}{2} \cos \beta_1 + C \sin\beta_1 + \frac{A + B}{2}. \label{even_e12}
 \end{align}  
By varying $\beta_1$, the minimum of $\arcsin|e_{12}(\beta_1)|$ is exactly given by: 
 \begin{equation}
 \min \left[ \arcsin|e_{12}(\beta_1)|\right]=\Lambda(\theta,\psi,\phi,\Theta). \label{arcf}
 \end{equation}
 Combining Eq.~\eqref{arce12} and Eq.~\eqref{arcf}, one can conclude that (i) of Theorem \ref{th:even} holds true.

\textbf{Case (ii)}:

By taking the inverse operation on both sides of Eq.~\eqref{even2}, it is equivalent to 
 \begin{align}
	 R \left( \theta,\psi,-\phi \right) =& R\left( \hat{m}, -\gamma_l \right) R \left(\hat{z},-\beta_l \right) \ldots R\left(\hat{m}, -\gamma_1 \right)R(\hat{z},-\beta_1),
	 \end{align}
	 or
	  \begin{align}
	 R \left( \theta,\psi,\tilde{\phi} \right)=&R\left( \hat{m}, \tilde{\gamma}_l \right) R \left(\hat{z},\tilde{\beta}_l \right) \ldots R\left(\hat{m}, \tilde{\gamma}_1 \right)R(\hat{z},\tilde{\beta}_1),\label{even2_2}
\end{align}
where 
\begin{subequations}
\begin{align}
\tilde{\phi}&=-\phi~\rm{mod}~4\pi,\\
\tilde{\gamma}_i&=-\gamma_i~\rm{mod}~4\pi,\\ 
\tilde{\beta}_i&=-\beta_i~\rm{mod}~4\pi.
\end{align}\label{tilde}
\end{subequations}
 According to (i) of Theorem \ref{th:even}, the existence of Eq.~\eqref{even2_2} is equivalent to 
 \begin{equation}
 \Lambda(\theta,\psi,\tilde{\phi},\Theta)\leqslant (l-1)\Theta.
 \end{equation}
  Since $ \Lambda(\theta,\psi,-\phi,\Theta)=\Lambda(\theta,\psi,\tilde{\phi},\Theta)$, (ii) of Theorem \ref{th:even} also holds true.
 \end{proof}

 \subsubsection{Constructing even-piece decompositions}\label{sec:even}
 When Eq.~\eqref{even_cr1} is satisfied, we should construct Eq.~\eqref{even1}. Firstly, $\beta_1$ can take any values that satisfy Eq.~\eqref{arce12}, or one can simply take
 \begin{equation}
 \beta_1=\pi+\text{Arg}\left(\frac{B-A}{2}+iC \right),
 \end{equation}
 which makes the left hand side of Eq.~\eqref{arce12} reach its minimum. If we denote $R(\theta',\psi',\phi')=R(\hat{m},-\beta_1)U(\theta,\psi,\phi)$, 
 other parameters can be obtained by applying the odd-piece decomposition scheme (cf. Sec.~\ref{sec:odd}) on $R(\theta',\psi',\phi')$.
 
When Eq.~\eqref{even_cr2} is satisfied, we should construct Eq.~\eqref{even2}. To do so, we construct Eq.~\eqref{even2_2} first, then determine the values of $\tilde{\beta}_{i}, \tilde{\gamma}_{i}$ with the same method of (i). After that, one can obtain the values of $\beta_i,\gamma_i$ in Eq.~\eqref{even2} from Eq.~\eqref{tilde}.
 
\subsection{Minimum number of pieces for all possible rotations} \label{sec:min}
We separate the problem into two cases: odd-piece and even-piece decompositions. We recall that $\Theta\in(0,\pi/2]$ is defined as the angle between two fixed axes.  ~\\
\textbf{Theorem 6.1} \textit{ (Odd-piece) Arbitrary rotations $U\in\mathcal{A}$ can be decomposed to $2l-1$ pieces with $l\in \mathbb{Z}^{+}$, if and only if 
 \begin{equation}
 \Theta\geqslant \frac{\pi}{2(l-1)}. \label{eq:crodd}
 \end{equation}
\label{thm:supplodd}}
\begin{proof}

\textit{Sufficiency}:

When Eq.~\eqref{eq:crodd} is satisifed, we have $(l-1)\Theta\geqslant\frac{\pi}{2}$. Since $|\delta_1(\theta,\phi)|\in[0,\frac{\pi}{2}]$ and $|\delta_2(\theta,\psi,\phi,\Theta)|\in[0,\frac{\pi}{2}]$, according to Theorem \ref{th:odd}, there exist a decomposition of $2l-1$ pieces  for arbitrary rotations.

\textit{Necessity}:

For $R(\frac{\pi}{2},\frac{\pi}{2},\pi)$, we have 
\begin{equation}
	 \delta_1\left(\frac{\pi}{2},\pi\right)= \delta_2\left(\frac{\pi}{2},\frac{\pi}{2},\pi,\Theta\right) =\frac{\pi}{2}.
\end{equation}
So when $\Theta< \frac{\pi}{2l}$, neither Eq.~\eqref{crodd_1} nor Eq.~\eqref{crodd_2} can be satisfied. According to Theorem \ref{th:odd}, the $2l-1$ pieces decomposition of $R(\frac{\pi}{2},\frac{\pi}{2},\pi)$ does not exist.
\end{proof}~\\~\\
\textbf{Theorem 6.2} \textit{(Even-piece) Arbitrary rotations $U\in\mathcal{A}$ can be decomposed to $2l$ pieces, if and only if 
 \begin{equation}
 \Theta\geqslant \frac{\pi}{2l-1}.\label{creven}
 \end{equation}}

\begin{proof}

\textit{Sufficiency}:
 
   According to Lemma \ref{le:two}, arbitrary rotations can be written as 
 \begin{equation}
 U(\theta,\psi,\phi)=R(\hat{z},2\psi)R(\theta'',0,\phi''), \label{xz_1}
 \end{equation}
  for certain values of $\theta''\in[0, \pi)$, and $\phi''\in[0,4\pi)$. To show the existence of $2l$-piece decomposition for $U(\theta,\psi,\phi)$, one only needs to prove that $R(\theta'',0,\phi'')$ can always be decomposed into $2l-1$ pieces when Eq.~\eqref{creven} is satisfied.
 
For $R(\theta'',0,\phi'')$  we have
  \begin{subequations}
 \begin{align}
 \delta_1(\theta'',\phi'')&= \arcsin\left| \sin\theta''\sin\frac{\phi''}{2}   \right|\leqslant |\theta''|,\\
	\delta_2(\theta'',0,\phi'',\Theta)&= \arcsin \left|\sin \frac{\phi''}{2} \sin \left(\Theta-\theta'' \right) \right|\leqslant |\Theta-\theta''|.
 \end{align}
 \end{subequations} 
It is easy to check that when $\theta''\in[0,(l-1)\Theta]\cup[\pi-(l-1)\Theta, \pi)$, we have
 \begin{equation}
 \delta_1(\theta'',\phi'')\leqslant (l-1)\Theta,
 \end{equation}
 and when 
 $\theta''\in[\Theta,l\Theta]$, we have
 \begin{equation}
 \delta_2(\theta'',0,\phi'',\Theta) \leqslant (l-1)\Theta.
  \end{equation}
  Therefore, when 
  \begin{equation}
  \theta''\in \mathcal{A}_{\theta}=[0,l\Theta]\cup[\pi-(l-1)\Theta,\pi],\label{cover}
  \end{equation}
   we have   
  \begin{equation}
  \Theta \geqslant  \frac{\min\left\{\delta(\theta'',\phi''),\delta'(\theta'',0,\phi'',\Theta)\right\} }{l}.
  \end{equation}
  According to Theorem \ref{th:odd}, $R(\theta'',0,\phi'')$ can be decomposed to $2l-1$ pieces when $\theta''\in \mathcal{A}_{\theta}$. Furthermore, for $\Theta\geqslant \frac{\pi}{2l-1}$, we have $[0,\pi)\subset\mathcal{A}_{\theta}$, so $R(\theta'',0,\phi'')$ can always be decomposed to $2l-1$ pieces. ~\\
    
\textit{Necessity}:

The necessity can be proven by finding specific rotations that fail to be decomposed to $2l$ pieces when $ \Theta< \frac{\pi}{2l-1}$.

For $l=1$, we consider $R(\frac{3\pi}{4},0,\pi)$, and notice that for arbitrary values of $\Theta\in(0,\frac{\pi}{2}]$, we have $\Lambda(\frac{3\pi}{4},0,\pm\pi,\Theta)=\frac{1}{4}\left[1+\sin^2 (2\Theta)\right]>0$. Therefore, $R(\frac{3\pi}{4},0,\pi)$ cannot be decomposed in two steps with $\hat{z}$ and $\hat{m}$.

 For $l>1$, we consider the rotation $R(\frac{l}{2l-1}\pi,0,\pi)$. When $\Theta< \frac{\pi}{2l-1}$, it is easy to check that  
\begin{equation}
\Lambda\left(\frac{l}{2l-1}\pi,0,\pm\pi,\Theta\right)=\frac{\pi}{2l-1}(l-1) > (l-1)\Theta.
\end{equation}
So according to Theorem \ref{th:even}, the rotation $R(-\frac{l-1}{2l-1}\pi,0,\pi)$ cannot be decomposed in $2l$ steps with $\hat{z}$ and $\hat{m}$.
 \end{proof}~\\~\\
Combining the above results, we have the following:

\begin{theorem}\label{th:fixed}

 Arbitrary rotations $U\in\mathcal{A}$ can be decomposed to $p$ pieces in the form of
 \begin{equation}
 U=\prod_{i=1}^{p}R(\hat{n}_i,\phi_i)
 \end{equation}
 with $R(\hat{n}_i,\phi_i)\in \mathcal{G}_{b}$ and $\hat{n}_i\neq\hat{n}_{i+1}$, if and only if 
 \begin{equation}
 \Theta\geqslant \frac{\pi}{p-1}.\label{creven}
 \end{equation}
\end{theorem}

\section{Code for constructing exact minimal decomposition sequence}\label{sec:code}

We have provided matlab code (Type\_I.m and Type\_II.m) for constructing explicit minimal decomposition of Type I and Type II qubits described in this work. 

For a target unitary transformation $U=U(\theta,\psi,\phi)$ and $\Theta$ (we restrict $\Theta\in(0,\pi]$ for Type I or $\Theta\in(0,\pi/2]$ for Type II), the inputs of the function are $\theta, \psi,\phi,\Theta$ respectively. For example, one inputs:
\begin{equation}
 >>\text{Type}\_\, \text{I}\;(\text{pi}/3,0,\text{pi},\text{pi}/4)\notag
\end{equation}

the output should be 
\begin{align}
    &0.7854 \quad 0 \quad\quad\quad\;\, 0.7854\notag\\
    &4.4411 \quad  0.7495 \quad  10.7243\notag
\end{align}
The first line corresponds to the polar angles of each elementary rotations in order while the second line corresponds to the rotation angle. One can verify that: 
\begin{widetext}
\begin{equation}
U(\pi/3,0,\pi)=R(0.7854,0,4.4411)R(0,0,0.7495)R(0.7854,0,10.7243).
\end{equation}
\end{widetext}
\end{appendix}


\begin{thebibliography}{43}
\bibitem [{rot()}]{rotnote}%
  \BibitemOpen
  \href@noop {} {}\bibinfo {note} {If available elementary rotation axes line
  in different range, one can make a coordinate transformation such that they
  coincide with $\mathcal{G}_\xi$ Type-I or $\mathcal{G}_b$ for Type-II. For
  example, $U\rightarrow R(\hat{y},\frac{\pi}{2}-\Theta) U
  R(\hat{y},-\frac{\pi}{2}+\Theta)$ brings $\mathcal{G}_{ZM}$ to
  $\mathcal{G}_\xi$, and $U\rightarrow R(\hat{y},-\frac{\pi}{3}) U
  R(\hat{y},\frac{\pi}{3})$ brings $\mathcal{G}_{EO}$ to $\mathcal{G}_b$ with
  $\Theta=\frac{\pi}{3}$.}\BibitemShut {Stop}%
\end{thebibliography}
\end{document}